\documentclass{IEEEtran}

\newcommand{\cJ}{\mathcal J}

\newcommand{\cL}{\mathcal L}

\newcommand{\cN}{\mathcal N}
\newcommand{\cO}{\mathcal O}
\newcommand{\cP}{\mathcal P}
\newcommand{\cQ}{\mathcal Q}
\newcommand{\cR}{\mathcal R}

\newcommand{\bbE}{\mathbb{E}}

\newcommand{\bbR}{\mathbb{R}}

\newcommand{\vc}{\bm{c}}
\newcommand{\vd}{\bm{d}}

\newcommand{\vg}{\bm{g}}

\newcommand{\vu}{\textbf{u}}

\newcommand{\vx}{\textbf{x}}

\newcommand{\bmat}{\begin{bmatrix}}
\newcommand{\emat}{\end{bmatrix}}

\newcommand{\bsmat}{\begin{bsmallmatrix}}
\newcommand{\esmat}{\end{bsmallmatrix}}

\usepackage{bm}
\usepackage{commath}
\usepackage{comment}
\usepackage{mathtools}
\usepackage{algorithm}
\usepackage[noend]{algpseudocode}
\usepackage{amsmath,amssymb,amsfonts,mathrsfs}
\usepackage{amsthm}
\allowdisplaybreaks

\usepackage{float}
\usepackage{xcolor}
\usepackage{dcolumn}
\usepackage{caption}
\usepackage{makecell}
\usepackage{multirow}
\usepackage{graphicx}
\usepackage{subcaption}
\usepackage{tabularray}
\usepackage{stackengine}

\usepackage[utf8]{inputenc}
\usepackage{url}
\usepackage{cite}
\usepackage{soul}
\usepackage{bigfoot}
\usepackage{textcomp}
\usepackage{enumitem}
\usepackage{anyfontsize}
\usepackage[euler]{textgreek}
\usepackage[hidelinks]{hyperref}
\DeclareNewFootnote[para]{default}

\theoremstyle{plain}
\newtheorem{lemma}{Lemma}
\newtheorem{remark}{Remark}

\newtheorem{theorem}{Theorem}
\newtheorem{corollary}{Corollary}
\newtheorem{proposition}{Proposition}

\theoremstyle{definition}
\newtheorem{example}{Example}
\newtheorem{problem}{Problem}
\newtheorem{definition}{Definition}
\newtheorem{assumption}{Assumption}

\theoremstyle{remark}

\def\BibTeX{{\rm B\kern-.05em{\sc i\kern-.025em b}\kern-.08em
    T\kern-.1667em\lower.7ex\hbox{E}\kern-.125emX}}

\begin{document}
\title{Inverse Linear Quadratic Gaussian Games: Constrained Setting and Transferability}

\author{Kai Ren, and Maryam Kamgarpour	\thanks{Ren and Kamgarpour are with the SYCAMORE Lab, École Polytechnique Fédérale de Lausanne (EPFL), Switzerland (e-mail: {\tt\small kai.ren@epfl.ch; maryam.kamgarpour@epfl.ch}).}
\thanks{Ren and Kamgarpour's research is supported by Swiss National Foundation Grant $\#200020\_207984 \slash  1$.}} 

\maketitle
\begin{abstract}    
This work addresses finite-horizon inverse linear quadratic Gaussian games. In a constrained setting, we characterize the set of cost parameters and optimal dual values that generate a given generalized Nash equilibrium, and we propose an algorithm to compute these parameters. In an unconstrained setting, we address transferability, namely, we bound the cost value perturbation between two policies: one induced by the identified cost parameters, the other by the expert parameters, under a set of different dynamics. This cost value perturbation scales linearly with the deviations in the dynamics and the identified cost parameters. Through numerical simulations, we show that, in a constrained setting, our algorithm identifies the cost parameters and dual values that can reproduce the policy and trajectories corresponding to the observed generalized Nash equilibrium. In an unconstrained setting, we show with a traffic simulation and real-robot experiments that the identified cost parameters can be used to control sufficiently close dynamics with performance degrading linearly with the deviations in the dynamics and the identified cost parameters.
\end{abstract}






\section{Introduction}

\IEEEPARstart{A}{n} optimal control problem relies on defining a cost metric and often it is assumed to be known; However, in practice, the cost functions are unknown, hard to specify by hand and only historical data are available. Inverse optimal control aims to recover the cost functions from observed policy or trajectory data. This problem is also known as cost identification. The identified cost function enables us to reproduce desired behaviour and adapt to different dynamics or environments. In this work, we address multi-agent finite-horizon inverse linear quadratic Gaussian games. 

In multi-agent systems such as warehouse robots and autonomous vehicles, multiple agents jointly control a dynamical system over a time horizon. The decision of one agent at one time step influences the cost values of the other agents \cite{Basar1998}. Cost identification in this setting is important for inferring agents' intents from interaction data and for reconstructing policies that respect shared constraints such as collision avoidance. Hence, a relevant problem is to identify cost functions in constrained dynamic games. Once the cost parameters are identified, they are typically used to control similar systems, such as different vehicle models. This raises the question of how much performance degrades when identified cost parameters are transferred to perturbed dynamics. To this end, we also study a notion of transferability \cite{Ng2000} in the multi-agent setting. Based on the above two problems introduced, we provide a review of the most relevant work. 

Inverse optimal control in a single-agent setting has been studied in both continuous-time~\cite{Kalman1964WhenIA, Priess2015} and discrete-time~\cite{Menner2018ECC} linear quadratic systems. For a given policy, the corresponding cost matrices are characterized through the algebraic Riccati equation, and their computation is formulated as a semidefinite program. Cost identification from the state \cite{ZHANG2019IOC} or output \cite{Zhang2019ILQR} trajectories has been addressed in a discrete-time finite-horizon inverse linear quadratic control. These works formulate a convex program from Pontryagin’s Maximum Principle characterization of the optimal control to compute the cost matrices. 

In multi-agent dynamic games, the characterization and computation of the cost matrices that induce a Nash equilibrium policy have been studied in both infinite- \cite{Inga2019, Huang2022} and finite-horizon \cite{ren2025identify} linear quadratic games. Across these works, the characterizations rely on coupled Riccati equations. Beyond the linear quadratic setting, cost identification in nonlinear dynamic games from open-loop Nash equilibria \cite{Molloy2017, peters2021inferringobjectives} has been addressed via maximum likelihood estimation (MLE) in continuous- and discrete-time dynamic games, respectively. However, all the above works focus on the unconstrained settings. 

The generalized Nash equilibrium addresses the equilibrium concept under the constrained setting. Generalized Nash equilibria are typically characterized through optimality conditions, such as the Karush–Kuhn–Tucker (KKT) conditions \cite{Zhu2023, Laine2023}. Accordingly, previous works addressed cost identification by minimizing the residual of the necessary conditions associated with an open-loop generalized Nash equilibrium in constrained differential games \cite{Awasthi2020}. A KKT-constrained MLE method has also been used to identify cost parameters from noisy state and input trajectories corresponding to feedback~\cite{Li2023} or open-loop~\cite{Liu2023} generalized Nash equilibria. In parallel, Bayesian methods have been developed for cost identification under constraints, where Gaussian \cite{cleach2020lucidgamesonlineunscentedinverse} or discrete \cite{peters2023contingencygames} distributions over the cost parameters are updated using online state and input observations. While effective in practice, neither approach characterizes the cost parameters exactly: MLE methods rely on approximated gradients of the forward game \cite{Awasthi2020, Li2023, Liu2023}, and Bayesian methods rely on iterative updates from an initial guess of the distribution over the cost \cite{cleach2020lucidgamesonlineunscentedinverse, peters2023contingencygames}. An exact characterization of the cost parameters corresponding to generalized Nash equilibria remains open, even in a constrained linear quadratic games. 

Our first contribution is to  characterize the set of cost parameters and the optimal dual values corresponding to a given generalized Nash equilibrium. We also propose an algorithm to compute them. 

Once the cost parameters are identified, the next step is to apply them to control or predict the trajectories of systems whose dynamics are close to the experts' system that has been used to solve the inverse problem. The notion of transferability \cite{Ng2000} bounds the cost value perturbation between two policies: one induced by the identified cost parameters on a different system, the other by the expert cost parameters on that system.

Given error in dynamics, past works have characterized the cost value perturbation in linear quadratic control. In the single-agent linear quadratic regulator setting with given cost matrices and an $\epsilon$ dynamics error, the perturbation on the cost value scales as $\cO(\epsilon)$ ~\cite{Fiechter1997, Dean2020}. For sufficiently small $\epsilon$, an $\cO(\epsilon^2)$ perturbation bound on the cost value was derived \cite{Mania2019} by exploiting a second-order expansion of the cost around the optimal controller. However, these works focus on a single-agent infinite-horizon settings with stationary dynamics and cost parameters. Moreover, the cost value perturbation bounds account only for errors in the system dynamics. 

In practice, the identified cost parameters themselves are often perturbed, especially when they are learned from noisy state and input demonstrations \cite{ren2025identify, peters2021inferringobjectives}. Starting from perturbed cost matrices in a single-agent linear quadratic setting, the stationary Riccati solution satisfies $\|P - \hat P\| = \cO(\epsilon)$ under sufficiently small errors $\epsilon$ in both the dynamics and the cost matrices~\cite{Mania2019, Konstantinov1993}. However, these works bound the perturbation of the Riccati solution, rather than the resulting cost values. 

Our second contribution is to derive a transferability bound of the unconstrained finite-horizon LQG game under errors in the cost parameters and the dynamics. In a constrained setting, deriving the transferability bound requires a sensitivity analysis of the optimization problem for computing the optimal dual values. This is non-trivial, because the characterization of the optimal dual values admits multiple solutions and is not well-posed for the senitivity analysis. Hence, we focus on the transferability analysis in an unconstrained setting.

In summary, our contributions are as follows 
\begin{itemize}
    \item For a finite-horizon constrained LQG game, we characterize the set of cost parameters and the optimal dual variables corresponding to the generalized Nash equilibrium (see Proposition~\ref{prop:characterization}), and we develop an algorithm to compute them (see Algorithm~\ref{alg:constrainedILQG}).
    \item For the inverse LQG games in an unconstrained setting, we derive a transferability bound on the cost value when the identified cost parameters with errors are applied to perturbed multi-agent dynamics (see Theorem~\ref{theorem:transferability}).
    \item We validate our algorithm and the theoretical results in simulations and real-robot experiments (see Section~\ref{sec:caseStudy}). We show that our algorithm can identify cost parameters that recover the generalized Nash equilibrium policies in constrained LQG games. We also show that, under an unconstrained setting, the identified cost parameters are transferable to sufficiently close dynamics.
\end{itemize}
\textit{Notations:} A Gaussian distribution with mean $\mu$ and covariance matrix $\Sigma$ is denoted as $\mathcal{N}(\mu, \Sigma)$. We denote a set of consecutive integers by $[a]=\{1, 2, \dots, a\}$ and $[a]^- = \{0, 1, \dots, a-1\}$.  The weighted \(L_2\) norm of a vector \( x \) with a \( P \succeq 0\) is $\|x \|_{P} = \sqrt{x^\top P x}$. Let $I_n$ denote the identity matrix in $\bbR^{n\times n}$. 
We use \(\operatorname{diag}(A_1,\ldots,A_m)\) to denote the block-diagonal matrix with diagonal blocks \(A_1,\ldots,A_m\). The vectorization of $A \in \mathbb{R}^{m \times n}$ is denoted by $\operatorname{vec}(A) \in \mathbb{R}^{mn}$. The Kronecker product of $A \in \mathbb{R}^{m \times n}$ and $B \in \mathbb{R}^{p \times q}$ is denoted by $A \otimes B \in \mathbb{R}^{mp \times nq}$. Let $f(\cdot), g(\cdot)$ be real-valued functions, we write $f(\epsilon) = \mathcal{O}(g(\epsilon))$ if there exist constants $c > 0$ and $\epsilon_0$ such that $|f(\epsilon)| \le c\,|g(\epsilon)|$ for all $\epsilon \geq \epsilon_0$. 

\section{Inverse LQG game in constrained settings} \label{sec:constrained} 
We consider a finite-horizon discrete-time constrained linear quadratic Gaussian (LQG) game where the $N$-agent linear dynamics is \vspace{-3mm}
\begin{equation}
     x_{t+1} = A_t x_t + \sum_{i=1}^{N} B_{t}^{i} u^i_t + \omega_t, \label{eq:dynamics} 
\end{equation} 
In the above, the control input of agent \(i\) at time \(t\) is \(u_t^i\in\mathbb{R}^{n_u}\), and the concatenated state vector \(x_t=[x_t^{1\top},\ldots,x_t^{N\top}]^\top\in\mathbb{R}^{n_x}\) contains the states of all agents. The system dynamics matrices are $A_t \in \mathbb{R}^{n_x \times n_x}$ and $B_t^i \in \mathbb{R}^{n_x \times n_u}$. 
We consider a Gaussian initial state $x_0 \sim \cN(\mu_0, \Sigma_0)$ and Gaussian system noise $w_t \in \mathbb{R}^{n_x}$, with $w_t \sim \cN(0, W_t)$. We denote the expected state trajectory as $\bbE\{\vx\} = [\bbE(x_0)^\top, \; \bbE(x^{\top}_{1}), \ldots, \bbE(x^{\top}_{T})]^\top$.  \vspace{1mm}

Each agent $i$ aims to minimize their own cost function 
{\small
\begin{equation}
    \cJ^{i}(K, \alpha):= \bbE \left[\sum_{t=0}^{T-1} \left( \|x_{t+1}\|^2_{Q^{i}_{t+1}} + l^{i \top}_{t+1} x_{t+1} + \|u^i_t\|^2_{R^{i}_{t}} \right)\right],  \label{eq:LQcost}
\end{equation}}

\noindent
The cost matrices are $Q^{i}_{t} \in \mathbb{R}^{n_x \times n_x}$, $l^{i}_{t} \in \mathbb{R}^{n_x}$, and $R^{i}_{t} \in \mathbb{R}^{n_u \times n_u}$ with $Q^{i}_{t}\succeq 0, \; R^{i}_{t}\succ 0$. We consider an affine state-feedback policy for player \(i\) with \(u_t^i = -K_t^i x_t - \alpha_t^i\). Let us define \(K^i := \operatorname{diag}(K_0^i,\ldots,K_{T-1}^i,0_{n_x \times n_u}) \in \mathbb{R}^{(T+1)n_x \times (T+1)n_u}\) and \(\alpha^i := [(\alpha_0^i)^\top,\ldots,(\alpha_{T-1}^i)^\top,0_{n_u \times 1}]^\top  \in \mathbb{R}^{(T+1)n_u}\) as a compact notation of the policy parameters of player $i$ over the horizon. We denote the concatenated policy across agents by $K = (K^i, K^{-i})$ and $\alpha = (\alpha^i, \alpha^{-i})$. Each agent optimizes over its feedback parameters to minimize its cost \eqref{eq:LQcost}. 

We also consider the following affine constraints on the expected state trajectory that are shared across all agents:
\begin{equation}
    \vg(\bbE\{\vx\}) := \vc^{\top} \bbE\{\vx\} + \vd  \leq 0. \label{constraint:primalAffine}
\end{equation}
with $\vc \in \bbR^{T n_x \times O}$ and $\vd \in \bbR^O$, where $O$ is the total number of constraints. 

Overall, each agent $i$ aims to solve the following problem:
{\small
\begin{subequations} \label{problem:primal} 
\begin{alignat}{2} 
\underset{K^{i}, \alpha^{i}}{\text{min}} & \quad  \cJ^{i}(K, \alpha)\\
\text{s.t. }
& \quad x_{t+1} = A_t x_t + \sum_{i=1}^{N} B_{t}^{i} u^i_t + \omega_t, \\
& \quad \vg(\bbE\{\vx\}) := \vc^{\top} \bbE\{\vx\} + \vd  \leq 0, \label{eq:primalAffine} \\
& \quad x_0 \sim \cN(\mu_0, \Sigma_0), \quad w_t \sim \cN(0, W_t). \label{eq:Gaussian}
\end{alignat}
\end{subequations}}

To motivate the setting above, we consider the following example throughout the development of this work. 

\begin{example}
\textit{Consider a two-car interaction over a finite horizon. The state $x^i_t$ represents the 2D position of car $i$. Each car aims to follow a desired reference $\bar x_t^i$. The tracking is captured by a quadratic cost defined as}
\[\bbE\left[ \sum_{t=1}^T \|x_t^i-\bar x_t^i\|_2^2 \right] = \bbE\left[\sum_{t=1}^T x_t^{i\top} Q_t^i x_t^i + l_t^{i\top} x_t^i + c_t^i\right],\]
\textit{where $Q_t^i = I$ and $l_t^i = -2\bar x_t^i$.
A collision avoidance chance constraint is defined as
\(\Pr(\|x^1_t - x^2_t\|_2 \ge r) \ge 1-\delta, \;\; \forall t \in [T] \). It is known that this constraint can be conservatively approximated as \eqref{constraint:primalAffine} \cite{Ren2025,nair2023predictive}. This gives rise to an LQG game \eqref{problem:primal} with a linear term in the cost encoding goal-tracking information, whereas the shared affine constraint encoding collision avoidance.}
\end{example}

\begin{definition}
A generalized Nash equilibrium is a feasible joint control policy $(K^{*}, \alpha^{*})$ where no player has a feasible unilateral deviation that can reduce their cost, that is, for all players $i \in [N]$,
\begin{equation*}
\begin{aligned}
   &\cJ^i(K^{*}, \alpha^{*}) \leq \cJ^i((K^i,\alpha^i), (K^{-i*}, \alpha^{-i*})), 
\end{aligned}
\end{equation*}
for all $(K^i,\alpha^i)$ such that the expected state trajectory $\bbE\{\vx\}$ induced by $((K^i,\alpha^i), (K^{-i*}, \alpha^{-i*}))$ satisfies \eqref{constraint:primalAffine}.
\end{definition} 

The generalized Nash equilibrium can be characterized by duality \cite{Facchinei2007, Ren2025}. Given that player \(i\) has cost \eqref{eq:LQcost} and constraints \eqref{constraint:primalAffine}, we denote player $i$'s Lagrangian by \(
    \cL^i((K, \alpha), \bm{\lambda}) := \cJ^i(K, \alpha) + {\bm{\lambda}}^\top \vg(\bbE\{\vx\}).\)
Because the constraints \eqref{constraint:primalAffine} are shared among all agents, we consider a class of games in which the dual variables ${\bm{\lambda}}  \in \bbR^O$ are also shared \cite{Zhu2023, cleach2021algames}. We make the following assumption for characterizing the generalized Nash equilibrium.

\begin{assumption}[\cite{Ren2025}]\label{assumption:slater}
For all players $i\in [N]$, for any $(K^{-i}, \alpha^{-i})$, the primal and dual problems of player $i$ have zero duality gap.
\end{assumption}

Under Assumption~\ref{assumption:slater}, the necessary and sufficient optimality conditions \cite{ruszczynski2006nonlinear,Facchinei2007} for the generalized Nash equilibrium \((K^{i*},\alpha^{i*})\) and the optimal dual value \(\bm{\lambda}^*\) are
\begin{align}
    & \cL^i((K^*, \alpha^*), \bm{\lambda}^{*}) = \min_{K^i,\alpha^i} \cL^i((K^i,\alpha^i), (K^{-i*}, \alpha^{-i*}), \bm{\lambda}^{*}), \notag \\
    & \bm{\lambda^{*}} \geq 0,  \quad\;\; \vg(\bbE\{\vx\}^*) \leq 0, \quad\;\; {\bm{\lambda^{*}}}^\top  \vg (\bbE\{\vx\}^*) = 0.  \label{eq:KKTconditions}
\end{align}

Note that under linear constraints \eqref{constraint:primalAffine}, the generalized Nash equilibrium of \eqref{problem:primal} is also the Nash equilibrium of an unconstrained LQG game with a cost that includes an auxiliary linear term \(\bar l^{i}_{t+1} = l^{i}_{t+1} + \lambda_{t+1}^{*\top} c_{t+1}^\top\):
\[
\cJ^{i}(K, \alpha) = \bbE \Big[\sum_{t=0}^{T-1} \big( \|x_{t+1}\|^2_{Q^{i}_{t+1}} + \bar l^{i \top}_{t+1} x_{t+1} + \|u^i_t\|^2_{R^{i}_{t}} \big)\Big].
\]

In this work, we study the following inverse problem of identifying cost parameters and the optimal dual values from given generalized Nash equilibrium and known shared constraints (e.g., collision avoidance).

\begin{problem} \label{problem:constIdentify}
    Given the multi-agent system dynamics \eqref{eq:dynamics}, the shared constraints \eqref{constraint:primalAffine} and the observed control policy $(K^{*}, \alpha^{*})$, identify cost parameters $\{Q^{i}_{t+1}, l^{i}_{t+1}, R^{i}_{t}\}_{t \in [T]^-}^{i \in [N]}$ and the dual value $\bm{\lambda}^*$, such that the generalized Nash equilibrium of \eqref{problem:primal} coincide with $(K^{*}, \alpha^{*})$. 
\end{problem}

For the rest of this section, we first characterize the cost parameters and the optimal dual value associated with the observed generalized Nash equilibrium in section~\ref{sec:unconstrained}. We then present a computation algorithm in section~\ref{sec:CIconstraint}.

\subsection{Characterization of the cost parameters and the optimal dual value} \label{sec:unconstrained}
Our appraoch to address problem \ref{problem:constIdentify} is based on the observation that that we can first solve for the quadratic cost parameters and the auxiliary linear term,
\begin{equation*}
\bar\theta_t^i :=
\big[
\textnormal{vec}(Q_{t}^i)^\top, \;
(\bar l_{t}^i)^\top,  \;
\textnormal{vec}(R_{t-1}^{i})^\top,  \;
1
\big]^\top \in \bbR^L,
\end{equation*}
based on past work's cost characterization on unconstrained LQG games \cite{Ren2025}. 

Next, to distangle the linear cost term $l_t^i$ from the auxiliary term $\bar l_t^i = l^{i}_{t+1} + \lambda_{t+1}^{*\top} c_{t+1}^\top$, we can leverage the charcterizatoin of optimal dual values from \eqref{eq:KKTconditions}.  To illustrate the importance of separating $l_t^i$ from the auxiliary linear term $\bar l_t^i$, we return to the example.

\noindent
\textbf{Example 1 (continued):} \textit{The linear cost term is defined as \(l_{t+1}^i = -2\bar x_{t+1}^i\). Separating this term gives the reference states \(\bar x_{t+1}^i\) for each car $i \in [N]$. In practice, this helps reveal whether an observed behavior, such as a sharp turn, is caused by the agent's tracking objective or by an active collision-avoidance constraint.} 



To characterize the quadratic cost parameters and the auxiliary linear terms, we first define the following backward recursion that is necessary for the derivations. 

At $t = T$, define \(\bar\Delta_T^i=0_{n_x^2}\) and \(\bar\Omega_T^i=0_{n_x}\), and the coupled Riccati terms are \(P_T^{i*}=Q_T^i\) and \(\zeta_T^{i*}=\bar l_T^i\). \vspace{1mm} 

At \(t=T-1,\ldots,1\), define
{\footnotesize
\begin{subequations} \label{eq:backRec}
\begin{alignat}{2}
&\bar\Delta_t^i=\textnormal{vec}\!\left(F_t^\top P_{t+1}^{i*}F_t+K_t^{i*\top}R_t^iK_t^{i*}\right), \\ 
&\bar\Omega_t^i=\textnormal{vec}\!\left[F_t^\top\!\left(\zeta_{t+1}^{i*}-P_{t+1}^{i*}\sum_{j=1}^NB_t^j\alpha_t^{j*}\right)+(K_t^{i*})^\top R_t^i\alpha_t^{i*}\right]. 
\end{alignat}
\end{subequations}}

\noindent
The Riccati terms \(P_t^{i*}\) and \(\zeta_t^{i*}\) are then defined by
{\footnotesize
\begin{subequations}
\begin{align}
&\textstyle P^{i*}_t = F_t^\top P^i_{t+1} F_t + (K^{i*}_t)^\top R_t^i K^{i*}_t  + Q_{t}^{i}, \label{recatti:P} \\
&\textstyle \zeta^{i*}_t = F_t^\top \left( \zeta^{i*}_{t+1} - P^i_{t+1} \sum_{j = 1}^N B_t^j \alpha_t^{j*} \right) + (K^{i*}_t)^\top R_t^i \alpha^{i*}_t + l^i_t. \label{recatti:zeta}
\end{align}
\end{subequations}}

To distangle the linear cost term $l_t^i$ from the auxiliary term $\bar l_t^i$, we need to characterize the optimal dual variable $\bm{\lambda}^*$ based on the optimality conditions \eqref{eq:KKTconditions}. To that end, a connection between the expected state trajectory and the dual variable is needed.
For the class of affine constraints, past work \cite{Ren2025} shows that \eqref{constraint:primalAffine} depends affinely on an arbitrary dual variable: 

\begin{equation}\label{eq:linearInLambda}
\vg(\bbE\{\vx\}_{\bm{\lambda}})=\vc^\top\bbE\{\vx\}_{\bm{\lambda}}+\vd=\tilde C\bm{\lambda}+\tilde d,
\end{equation}
where \(\tilde C:=\vc^\top ZHG\in\bbR^{O\times O}\) and \(\tilde d:=\vc^\top Fx_0+\vd\in\bbR^O\). The matrices \(F, Z, H, G\) are defined in Appendix~\ref{appendix:params} and they are determined by the known dynamics, the equilibrium policy, and the Riccati parameters.

Based on the definitions above, we now present the characterization of the cost parameters and the optimal dual values.

\begin{proposition} \label{prop:characterization}
For a given generalized Nash equilibrium policy \(\{K^{i*}_{t}, \alpha^{i*}_{t}\}_{t \in [T]^-}^{i \in [N]}\), the corresponding cost parameters and optimal dual value can be characterized as follows.

\textnormal{(i)} For all $i \in [N]$ and $t \in [T]$, the quadratic cost parameters and auxiliary linear term $\bar\theta_t^i$ is characterized by 
\begin{subequations} \label{eq:characUnconstrained}
\begin{alignat}{2}
     &M_t^i\bar\theta_t^i=0_{n_u(n_x+1)\times1}, \\
     &Q_{t}^i \succeq 0, \quad R_{t-1}^i \succ 0,
\end{alignat}
\end{subequations} 
where for all $i \in [N]$ and $t \in [T]$,

{\small
\begin{align*}
    M_t^i = \begin{bmatrix}
        S_{t-1}^i & 0_{n_un_x \times n_x} & - K_{t-1}^{i* \top} \otimes I_{n_u} & S_{t-1}^i\bar{\Delta}_{t}^i \\
        0_{n_u \times n_x^2} & B_{t-1}^{i \top} &  - E_{t-1}^{i \top} \otimes I_{n_u} &  B_{t-1}^{i \top}\bar{\Omega}_{t}^i
    \end{bmatrix}.
\end{align*}}

\noindent
Here, \(S_t^i = F_t^\top \otimes (B_t^i)^\top, \;\textstyle E_{t}^i = K_{t}^{i*} F_{t}^{-1} \sum_{j=1}^N B_{t}^j \alpha^{j*}_{t} + \alpha^{i*}_{t}\). The terms \(\bar\Delta_t^i\) and \(\bar\Omega_t^i\) are determined by the backward recursion defined in \eqref{eq:backRec}.

\textnormal{(ii)} The optimal dual variable \(\bm{\lambda}^*\) satisfies
\begin{subequations}  \label{eq:dualCharacterization}
    \begin{align}
	    &{\bm{\lambda^{*}}}^\top (\tilde{C} \bm{\lambda^*} + \tilde d) = 0, \\
	    &\bm{\lambda^*} \geq 0, \quad \text{and} \quad \tilde{C} \bm{\lambda^*} + \tilde d \leq 0. 
\end{align}
\end{subequations}
\end{proposition}

\textit{Proof.} \textnormal{(i)} The result follows directly from  the null-space characterization of the unconstrained LQG game \cite[Proposition 1]{ren2025identify}.

\textnormal{(ii)} 
Substituting \(\vg(\bbE\{\vx\}^*)=\tilde C\bm{\lambda}^*+\tilde d\) from \eqref{eq:linearInLambda} into \eqref{eq:KKTconditions} gives the characterization in \textnormal{(ii)}. \(\hfill \square\)

\subsection{Computation of the cost parameters and the optimal dual values} \label{sec:CIconstraint}
We aim to use Proposition~\ref{prop:characterization} to compute cost parameters and an optimal dual value consistent with the observed equilibrium policy. Both the null-space characterization \eqref{eq:characUnconstrained} and the dual variable characterization \eqref{eq:dualCharacterization} may admit multiple solutions. To select one non-trivial solution, we propose a quadratic program to compute the cost parameters and a linear complementarity problem to compute the optimal dual value.

Based on proposition~\ref{prop:characterization}\textnormal{-(i)}, we first compute the quadratic cost parameters that satisfy \(R_t^i \succ 0\), \(Q_t^i \succeq 0\) and auxiliary linear term \(\bar l_{t}^i\) by solving
\begin{subequations} \label{eq:QP}
    \begin{alignat}{2}
        \min_{\bar\theta^i_t} & \;\; r_t^i = \|M_t^i \bar\theta^i_t\|_2^2, \\
        \text{s.t. } & \;\; D_{1} \bar\theta_t^i \ge \tau, \quad \underline{\tau} \le D_{2} \bar\theta_t^i \le \bar{\tau}, \label{constraint:psd}
    \end{alignat}
\end{subequations}
where \(D_1 \in \mathbb{R}^{m_1 \times L}\) and \(D_2\in \mathbb{R}^{m_2 \times L}\) extract the diagonal and off-diagonal entries of \(R_{t}^i\) and \(Q_{t+1}^i\), respectively. In \eqref{constraint:psd}, \(\underline{\tau}\) is a small positive value, and \(\tau, \bar{\tau}\) are chosen to enforce \(R_t^i\), \(Q_t^i\) to be diagonally dominant and hence \(R_t^i \succ 0\), \(Q_t^i \succeq 0\) by Gershgorin circle theorem.

Based on proposition~\ref{prop:characterization}\textnormal{-(ii)}, we then compute the optimal dual variable \(\bm{\lambda}^*\) by solving the following linear complementarity problem:
\begin{subequations} \label{eq:LCPdual}
    \begin{alignat}{2}
        & \min_{\bm{\lambda}} \; \;\; \bm{\lambda^{\top}} (-\tilde{C} \bm{\lambda} - \tilde d), \label{eq:dualResidual}\\
        & \text{ s.t.} \quad\; \bm{\lambda} \geq 0, \quad \text{and} \quad \tilde{C} \bm{\lambda} + \tilde d \leq 0. \label{const:linearLCP}
    \end{alignat}
\end{subequations}

Solving \eqref{eq:LCPdual} yields \(\bm{\lambda^{*}}\), which then allows us to recover the linear cost terms \(l^{i}_{t+1}\) from \(\bar l^{i\top}_{t+1} = l^{i\top}_{t+1} + \lambda_{t+1}^{*\top} c_{t+1}^\top\), since the shared constraints are known. The constrained cost identification procedure is summarized in Algorithm~\ref{alg:constrainedILQG}.

\begin{algorithm}[t]
\begin{minipage}{0.5\textwidth}
\caption{Cost identification in LQG games}
\label{alg:constrainedILQG}
\begin{algorithmic}[1]
\State \textbf{Initialize:} (generalized) Nash equilibrium: $(K^{*}, \alpha^{*})$; \texttt{setting} $\in \{\texttt{constrained}, \texttt{unconstrained}\}$

\For{$t = T-1, ..., 1, 0$} \vspace{1mm}

\State Solve \eqref{eq:QP} $\rightarrow \{Q^{i}_{t+1}, \bar l_{t+1}^{i}, R_t^{i}\}^{i \in [N]}_{t \in [T]^-}$; \label{algStep:solveNE} \vspace{1mm}
\EndFor

\If{\texttt{setting} == \texttt{constrained}} \vspace{1mm}
    \State  Solve \eqref{eq:LCPdual} $\rightarrow \{l_{t+1}^{i}\}^{i \in [N]}_{t \in [T]^-}, \bm{\lambda}^*$; \label{step:dualID} \vspace{1mm}
    \State \Return $\{Q^{i}_{t+1}, l_{t+1}^{i}, R_t^{i}\}^{i \in [N]}_{t \in [T]^-}$, $\bm{\lambda}^*$
\Else \vspace{1mm}
    \State \Return $\{Q^{i}_{t+1}, \bar l_{t+1}^{i}, R_t^{i}\}^{i \in [N]}_{t \in [T]^-}$
\EndIf
\end{algorithmic}
\end{minipage}
\end{algorithm}

\begin{remark} \label{remark:constrained}
Problem \eqref{eq:LCPdual} is a linear complementarity problem\cite{Leonard2007}, as it minimizes the complementarity residual \eqref{eq:dualResidual} against linear inequality constraints \eqref{const:linearLCP}. Conditions for the existence and uniqueness of a solution to \eqref{eq:LCPdual} have been studied \cite{Leonard2007, Li2016}. A sufficient condition is that $-\tilde{C}$ be a $P$-matrix, i.e., a square matrix whose principal minors are all positive; this class includes positive-definite matrices. Verifying whether $-\tilde{C}$ is a $P$-matrix is challenging as it is generally neither symmetric nor sign-definite. Standard sufficient conditions for the $P$-matrix property do not readily apply. Based on our numerical experiments in section \ref{sec:caseStudy}, although $-\tilde{C}$ is often not positive definite, \eqref{eq:LCPdual} consistently admits a solution. A characterization of existence and uniqueness conditions is left for future work. 
\end{remark}




\section{Transferability of the identified cost parameters} \label{sec:transferability}
In general, the identified cost parameters are used to infer  policies for different dynamical systems (e.g., a different car model). Since cost identification depends on the dynamics, a natural question is whether the cost parameters learned for a given dynamics result in optimal performance when it is transferred to a new dynamics. This problem is referred to as transferability \cite{Ng2000, Fiechter1997}. For notational convenience, throughout this section we use $A:=\{A_t\}_{t\in[T]^-}$,  $B:=\{B_t^i\}_{t\in[T]^-}^{i\in[N]}$, $\theta:=\{\theta_t^i\}_{t\in[T]^-}^{i\in[N]}$ to denote the collections of the system dynamics and cost parameters over all time steps and agents.

\begin{definition} \label{def:nu-transferable}
The identified cost parameters $\widehat{\theta}^{id}$ are \emph{$\nu$-transferable} to a set of dynamics $\mathcal{P}$ if, for all players $i\in[N]$ and all dynamics $(A,B)\in\mathcal{P}$,
{\small
\begin{align*}
\mathcal{J}^i\!\left(\theta^E, (K^*, \alpha^*)_{(A,B,\theta^E)}\right)
-
\mathcal{J}^i\!\left(\theta^E, (K^*, \alpha^*)_{(A,B,\widehat{\theta}^{id})}\right)
\le \nu.
\end{align*}}

\noindent
where $\theta^E$ denotes the expert cost parameters that yield the observed Nash equilibrium under the true dynamics, and $(K^*, \alpha^*)_{(A,B,\theta)}$ denotes the Nash equilibrium induced by cost parameters $\theta$ under dynamics $(A,B)$.
\end{definition}

 Our goal in this section is to address the transferability for the finite-horizon LQG game.

\begin{problem}\label{problem:transferability}
Given the identified cost parameters $\widehat{\theta}^{id}$ from Algorithm~\ref{alg:constrainedILQG}, determine a set of dynamics $\mathcal{P}$ and a bound $\nu$ such that $\widehat{\theta}^{id}$ is $\nu$-transferable to $\mathcal{P}$.
\end{problem}

By Definition~\ref{def:nu-transferable}, addressing Problem~\ref{problem:transferability} requires bounding the perturbation of the cost \(\mathcal{J}^i\), which depends directly on the expected state and input trajectories. As shown in section~\ref{sec:unconstrained}, the expected state trajectory can be written as \(\bbE\{\vx\}_{\bm{\lambda}}=Fx_0+ZH(G \bm{\lambda})\). Hence, in the constrained LQG game, bounding the transferability gap requires bounding the dual estimation error \(\|\widehat{\bm{\lambda}}-\bm{\lambda}^*\|\) induced by \eqref{eq:LCPdual}. As discussed in Remark~\ref{remark:constrained}, bounding \(\|\widehat{\bm{\lambda}}-\bm{\lambda}^*\|\) is challenging, as it would require a sensitivity analysis of the linear complementarity problem \eqref{eq:LCPdual}. However, the matrix \(\tilde C\) in \eqref{eq:LCPdual} is not guaranteed to satisfy structural conditions, such as positive definiteness, that would allow for a sensitivity analysis \cite{Mangasarian1987}.

For this reason, we focus on analyzing the transferability in the unconstrained LQG game. Related questions have been studied in a single-agent infinite-horizon setting \cite{Fiechter1997, Dean2020, Mania2019}. In contrast, our setting is multi-agent, finite-horizon, time-varying. This setting is challenging because the errors propagate through the finite-horizon backward coupled Riccati recursions. 

 To obtain our transferability bound, we first bound the policy perturbation due to cost parameter error and dynamics error, respectively, in section~\ref{sec:policyErr}. We then propagate the policy perturbation to obtain the transferability bound in section~\ref{sec:transBound}.

\subsection{Policy perturbation under cost parameters and dynamic errors} \label{sec:policyErr}

In this section, we analyze the policy perturbations, as shown in Fig.~\ref{fig:policy-perturbation}, under errors on both the identified cost parameters and dynamics. 

In this section, we consider two dynamics: the expert's $(A^E,B^E)$, and the dynamics $(A',B')$ to which the identified cost is transferred. We also distinguish three cost parameters. The expert cost parameters $\theta^E$ induce the observed Nash policy under $(A^E,B^E)$. Since the optimization problem \eqref{eq:QP} used to compute the cost parameters has multiple solutions, let $\theta^{id}$ denote any identified cost parameters that constitute an optimal solution to \eqref{eq:QP} under $(A^E, B^E)$. As shown in Fig.~\ref{fig:policy-perturbation}, both $\theta^E$ and $\theta^{id}$ induce the same expert policy at $(A^E, B^E)$. Finally, since the cost identification itself is subject to error, let $\widehat{\theta}^{id}$ denote the perturbed identified parameters. We assume there exists an arbitrarily small \(\epsilon > 0\), such that
\begin{equation}\label{eq:costParamBound}
    \max_{i \in [N],\, t \in [T]^-}
    \left\{
    \| \widehat \theta_{t}^{id,i} - \theta_{t}^{id,i}\|
    \right\}
    \le \epsilon.
\end{equation} 
When the Nash policy used in Algorithm~\ref{alg:constrainedILQG} is estimated from sufficiently many trajectory samples, \eqref{eq:costParamBound} holds with high probability \cite[Theorem 1]{ren2025identify}.

In this subsection, we bound the deviation between the Nash policy induced by $\theta^E$ and by $\widehat{\theta}^{id}$ under $(A',B')$, that is the purple line in Fig.~\ref{fig:policy-perturbation}. This policy perturbation bound will be used to establish the transferability bound in the next subsection. To do so, we first bound the policy perturbation due to the cost-parameter error (orange dashed line) in Lemma~\ref{lemma:policyBound}. Next, Lemma~\ref{lemma:policyBoundDyn} bounds the policy perturbation due to the dynamics error (red dash-dotted lines). Finally, Proposition~\ref{remark:fullPolicyBound} obtains the desired policy perturbation bound via the triangle inequality.

First, we analyze the policy perturbation bound under small errors \eqref{eq:costParamBound} on the cost parameters.

\begin{figure}
    \centering
    \includegraphics[width=0.88\linewidth]{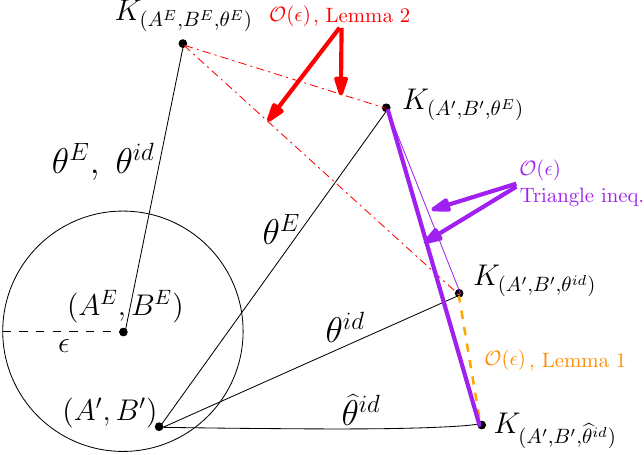}
    \caption{\textit{ The red dash-dotted line shows the policy perturbation due to the dynamics error alone, from $(A^E,B^E)$ to $(A',B')$ under \(\theta^E\). The orange dashed line shows the policy perturbation due to the cost-parameter error alone, from $\theta^{id}$ to $\widehat{\theta}^{id}$ under $(A',B')$. The purple line shows the policy perturbation from $\theta^E$ to $\widehat{\theta}^{id}$ under $(A',B')$.}}
    \label{fig:policy-perturbation}
\end{figure} 

\begin{lemma} \label{lemma:policyBound}
    For arbitrary multi-agent dynamics $(A,B)$ and for all $t \in [T]^-$, the multi-agent policies induced by $\widehat{\theta}^{id}$ and $\theta^{id}$ satisfy
\begin{align*}
& \|(K_t)_{(A,B,\widehat\theta^{id})} - (K_t)_{(A,B,\theta^{id})}\|
=
\cO\!\Big(
N^{(T-t)^2}\;
\xi^{\,2(T-t+1)^2}\;
\epsilon
\Big), \\
& \|(\alpha_t)_{(A,B,\widehat\theta^{id})} - (\alpha_t)_{(A,B,\theta^{id})}\|
=
\cO\!\Big(
N^{\frac{(T-t)^2}{2}}\;
\xi^{\,2(T-t+1)^2}\;
\epsilon
\Big),
\end{align*}
where $\xi > 1$ uniformly bounds the dynamics, cost and policy parameters, the initial-state mean and covariance, and the system-noise covariances:
{\small
\begin{equation} \label{eq:uniformNormB}
    \max_{\substack{i \in [N]\\ t \in [T]^-}}
    \left\{
    \begin{aligned}
        &\|A_t\|,\; \|B_t^i\|,\; \|W_t\|,\\
        &\|Q_{t+1}^i\|,\; \|l_{t+1}^i\|,\; \|R_t^i\|,\\
        &\|K_t\|,\; \|\alpha_t\|,\; \|\mu_0\|,\; \|\Sigma_0\|
    \end{aligned}
    \right\}\leq\xi.
\end{equation}}
\end{lemma}
Lemma~\ref{lemma:policyBound} shows that if the identified cost parameters incur an \(\epsilon\) error, then the Nash equilibrium policy induced from the identified cost parameters $\widehat{\theta}^{id}$ incurs an $\cO(\epsilon)$ error. The constant term grows exponentially with $T-t$ due to the multiplicative error propagation of the Ricatti solution $P_t$ and $\zeta_t$ over the finite horizon (e.g., $\|\Delta P_{t}\|$ depends on \(\|P_{t+1}\|\,\|\Delta P_{t+1}\|\)). 

Here, we provide a sketch of the proof. The complete proof is detailed in Appendix~\ref{appendix:lemma1Proof}.

\begin{proof}[Proof sketch]
We use \(\Delta(.)\) to denote the difference between parameters induced by $\widehat{\theta}^{id}$ and $\theta^{id}$. The proof starts by expressing the policy perturbations in terms of the Riccati solution perturbations based on \eqref{eq:Keq} an \eqref{eq:alphaeq}:
\begin{align*}
\|\Delta K_t\|
&= 
\cO\!\Big(
\sqrt{N}\xi^5\,\epsilon + \sqrt{N}\xi^3\,\|\Delta P_{t+1}\|
\Big),\\
\|\Delta\alpha_t\|
&=
\cO\!\Big(
\sqrt{N}\xi^4\,\epsilon
+
\sqrt{N}\xi^2\,\|\Delta P_{t+1}\|
+
\xi^2\,\|\Delta\zeta_{t+1}\|
\Big).
\end{align*}
Thus, the key step is to bound the backward Riccati perturbations \(\|\Delta P_{t+1}\|\) and \(\|\Delta \zeta_{t+1}\|\). 

The backward recursion starts from the terminal conditions
\(
\|P_T\|=\|Q_T\|\le \xi,\;
\|\zeta_T\|=\|l_T\|\le \xi,\;
\|\Delta P_T\|=\|\Delta Q_T\|\le \epsilon,\;
\|\Delta \zeta_T\|=\|\Delta l_T\|\le \epsilon.
\)
We first bound the nominal Riccati terms $\|P_t\|,\|\zeta_t\| $, since they enter the perturbation recursion multiplicatively. From \eqref{recatti:P} and \eqref{recatti:zeta},
the bound unrolls to
\(
\|P_t\|=\cO\!\big(\xi (N\xi^4)^{T-t}\big)
\) and 
\(
\|\zeta_t\|
=
\cO\!\big(
\xi\,(N\xi^4)^{T-t}
\big).
\)

We then propagate the perturbations bounds $\|\Delta P_t\|, \|\Delta\zeta_t\|$ backward. Substituting the bound on \(\|P_{t+1}\|\) and \(\|\zeta_{t+1}\|\), and unrolling from \(t=T\) gives
\[
\|\Delta P_t\|
=
\cO\!\Big(
\epsilon\,
(N^{3/2}\xi^9)^{T-t}\,
(N\xi^4)^{\frac{(T-t-1)(T-t)}{2}}
\Big),
\]



\[
\|\Delta \zeta_t\|
=
\cO\!\Big(
\epsilon\,
(N\xi^7)^{T-t}\,
(N\xi^4)^{\frac{(T-t-1)(T-t)}{2}}
\Big).
\]
Substituting the resulting bounds for \(\|\Delta P_{t+1}\|\) and \(\|\Delta\zeta_{t+1}\|\) back into the displays for \(\|\Delta K_t\|\) and \(\|\Delta \alpha_t\|\) gives the policy perturbation bounds in Lemma~\ref{lemma:policyBound}. 
\end{proof}

We now establish an analogous bound for policy perturbation under fixed cost parameters and dynamic errors (red dotted dashed lines in Figure~\ref{fig:policy-perturbation}).

\begin{lemma} \label{lemma:policyBoundDyn}
    Let the cost parameters $\theta$ be fixed. For any $(A',B')\in\mathcal{P}$ satisfying
\[
\max_{t\in[T]^-,\,i\in[N]}\{\|A'_t-A_t^E\|,\|B_t^{i\prime}-B_t^{i,E}\|\}\le \epsilon,
\]
the Nash equilibrium policies induced by $\theta$ on $(A',B')$ and by $\theta$ on the true dynamics $(A^E,B^E)$ satisfy
{\small
\begin{align*}
&\|(K_t)_{(A',B',\theta)} - (K_t)_{(A^E,B^E,\theta)}\|
 =
\cO\!\Big(
N^{\frac{(T-t+2)^2}{2}}\,
\xi^{2(T-t+2)^2}
\epsilon
\Big), \\
&\|(\alpha_t)_{(A',B',\theta)} - (\alpha_t)_{(A^E,B^E,\theta)}\|
 =
\cO\!\Big(
N^{\frac{(T-t+2)^2}{2}}\,
\xi^{2(T-t+2)^2}
\epsilon
\Big).
\end{align*}}
\end{lemma}

The proof is similar to Lemma~\ref{lemma:policyBound}: the dynamics perturbations are first propagated through the feedback equations and then through the backward Riccati recursion. The complete proof is given in Appendix~\ref{appendix:lemma2Proof}. Next, we upper-bound the policy perturbation from $\theta^E$ and $\widehat \theta^{id}$ for any dynamics $(A, B) \in \cP$.

\begin{proposition}\label{remark:fullPolicyBound}
    For any $(A,B)\in\mathcal{P}$ and $t \in [T]^-$, the policy induced by the true cost parameters $\theta^E$ and the policy induced by the perturbed identified cost parameters $\widehat{\theta}^{id}$, both under dynamics $(A,B)$, satisfy

{\small
\begin{align*}
\|(K_t)_{(A,B,\theta^E)}-(K_t)_{(A,B,\widehat{\theta}^{id})}\|
&=
\cO\!\Big(
N^{(T-t+2)^2}\xi^{2(T-t+2)^2}\epsilon
\Big), \\
\|(\alpha_t)_{(A,B,\theta^E)}-(\alpha_t)_{(A,B,\widehat{\theta}^{id})}\|
&=
\cO\!\Big(
N^{\frac{(T-t+2)^2}{2}}\xi^{2(T-t+2)^2}\epsilon
\Big).
\end{align*}}
\end{proposition}

\begin{proof}
Since both $\theta^E$ and $\theta^{id}$ satisfy the conditions \eqref{eq:characUnconstrained} at $(A^E,B^E)$, applying Lemma~\ref{lemma:policyBoundDyn} to both $\theta^E$ and $\theta^{id}$ and Lemma~\ref{lemma:policyBound} to the cost parameter perturbation, the triangle inequality gives, for any $(A,B)\in\mathcal{P}$ and for all $t \in [T]^-$,
{\small
\begin{align*}
&\|(K_t)_{(A,B,\theta^E)}-(K_t)_{(A,B,\widehat{\theta}^{id})}\| \\
\le & \|(K_t)_{(A,B,\theta^E)} - (K_t)_{(A^E,B^E,\theta^E)}\|
  \\
& \quad +\|(K_t)_{(A^E,B^E,\theta^{id})} - (K_t)_{(A,B,\theta^{id})}\|\\
& \qquad
+ \|(K_t)_{(A,B,\theta^{id})} - (K_t)_{(A,B,\widehat{\theta}^{id})}\|,
\end{align*}}

\noindent
and the stated bound follows by substituting the bounds from Lemma~\ref{lemma:policyBound} and Lemma~\ref{lemma:policyBoundDyn}. The bound on $\alpha_t$ follows analogously. 
\end{proof}

Proposition~\ref{remark:fullPolicyBound} bounds the distance between the expert policy and the policy induced by both perturbed dynamics and perturbed identified cost parameters (the bold purple line in Fig.~\ref{fig:policy-perturbation}). Next, we propagate the policy perturbation through the trajectory and cost expressions to obtain the transferability result. 


\subsection{Transferability of the identified cost parameters to neighbouring dynamics} \label{sec:transBound}

Based on the policy perturbation results from Proposition~\ref{remark:fullPolicyBound}, we now present the main results of the transferability of the identified cost parameters to neighbouring dynamics. 

\begin{theorem} \label{theorem:transferability}
    The identified cost parameters with errors $\widehat{\theta}^{id}$ are $\nu$-transferable to \(\mathcal P\) (Definition~\ref{def:nu-transferable}), where
{\small
\begin{align}\cP:=\left\{(\hat A, \hat B) \; \bigg| \; \max_{t\in[T]^-,\,i\in[N]} \{\|\hat A_{t} - A_{t}^E\|, \|\hat B_t^i - B_t^{E,i}\|\}\le \epsilon\right\} \label{eq:dynamicNeighbour}\end{align}
and the cost value perturbation bound $\nu =\cO\!\Big(\epsilon T^6(N\xi^2)^{\cO(T^2)}\Big).$}
\end{theorem}

Theorem~\ref{theorem:transferability} shows that the cost value perturbation scales linearly with the errors of the identified cost parameters and the dynamics $\epsilon$. However, the constant $\nu$ depends exponentially on the square of the horizon $T^2$, inherited from the policy perturbation bound in Proposition~\ref{remark:fullPolicyBound}. Later, we numerically show that large perturbations can prevent the identified cost from recovering the Nash equilibrium policies and trajectories.

Here, we provide a sketch of the proof. The complete proof is provided in Appendix~\ref{sec:cost-bound-derivation}.

\begin{proof}[Proof sketch]
Based on Proposition~\ref{remark:fullPolicyBound} and consider the worst-case error bound over \(t\in[T]^-\), we use
\(
\bar\Delta K
=
\cO\!\Big(N^{(T+2)^2}\xi^{2(T+2)^2}\epsilon\Big),
\quad
\bar\Delta\alpha
=
\cO\!\Big(N^{\frac{(T+2)^2}{2}}\xi^{2(T+2)^2}\epsilon\Big).
\) \vspace{2mm}

Then, we upper bound the state trajectory perturbations: \(
\|\Delta\vx\|
\le
\|\Delta F\|\,{\|\mu_0\|}
+
\|\Delta Z\|\,\|\bm{\alpha}\|
+
{\|Z\|}\,\|\Delta\bm{\alpha}\|.
\)

Using a conservative product-difference estimate for the block products defining \(F\) and \(Z\), for example \(\|\hat F_1\hat F_0-F_1F_0\|\le \|\hat F_1-F_1\|\,\|\hat F_0\|+\|F_1\|\,\|\hat F_0-F_0\|\), we obtain \(\|\Delta F\|=\cO(T^2\bar F_t^T\bar\Delta F_t)\) and \(\|\Delta Z\|=\cO(T^3\sqrt N\,\xi\bar F_t^T\bar\Delta F_t)\).

For the stacked inputs, since \(\vu^i=-K^i\vx-\alpha^i\), we have
\[
\|\Delta\vu^i\|
\le
\|\Delta K^i\|\,\|\vx\|
+
{\|K^i\|}\,\|\Delta\vx\|
+
\|\Delta\alpha^i\|
\]
We then substitute these state and input perturbations into the quadratic cost difference. {The covariance trace terms due to initial-state uncertainty and process noise also contribute \(\cO(\epsilon)\).} Collecting the dominant terms yields the stated transferability bound.     
\end{proof}

In practice, the Nash equilibrium policy used for identification is estimated from finite demonstrations rather than known exactly. Combining Theorem~\ref{theorem:transferability} with an existing sample complexity result for policy estimation yields the following guarantee.

\begin{corollary}\label{cor:sampleTransfer}
Suppose the Nash equilibrium policy used for cost identification in Algorithm~\ref{alg:constrainedILQG} is estimated from
\(
n = \mathcal{O}\!\left(\epsilon^{-2}\log\!\left(NT/\delta\right)\right)
\)
demonstration samples. Then, with probability at least \(1-\delta\), the identified cost parameters are \(\nu\)-transferable to \(\mathcal P\) as per Theorem~\ref{theorem:transferability}.
\end{corollary}
\begin{proof}
    \cite[Theorem 1]{ren2025identify} shows that \eqref{eq:costParamBound} holds with probability at least \(1-\delta\) with \(n \geq \mathcal{O}\!\left(\epsilon^{-2}\log\!\left(NT/\delta\right)\right)\) demonstration samples under conditions that active set of constraints in \eqref{constraint:psd} remains unchanged with the sample-estimated equilibrium. It follows directly from Theorem~\ref{theorem:transferability} that the stated sample complexity yields \(\nu\)-transferability. 
\end{proof} \vspace{-2mm}

\begin{figure}[t]
    \centering
    \includegraphics[width=0.98\linewidth]{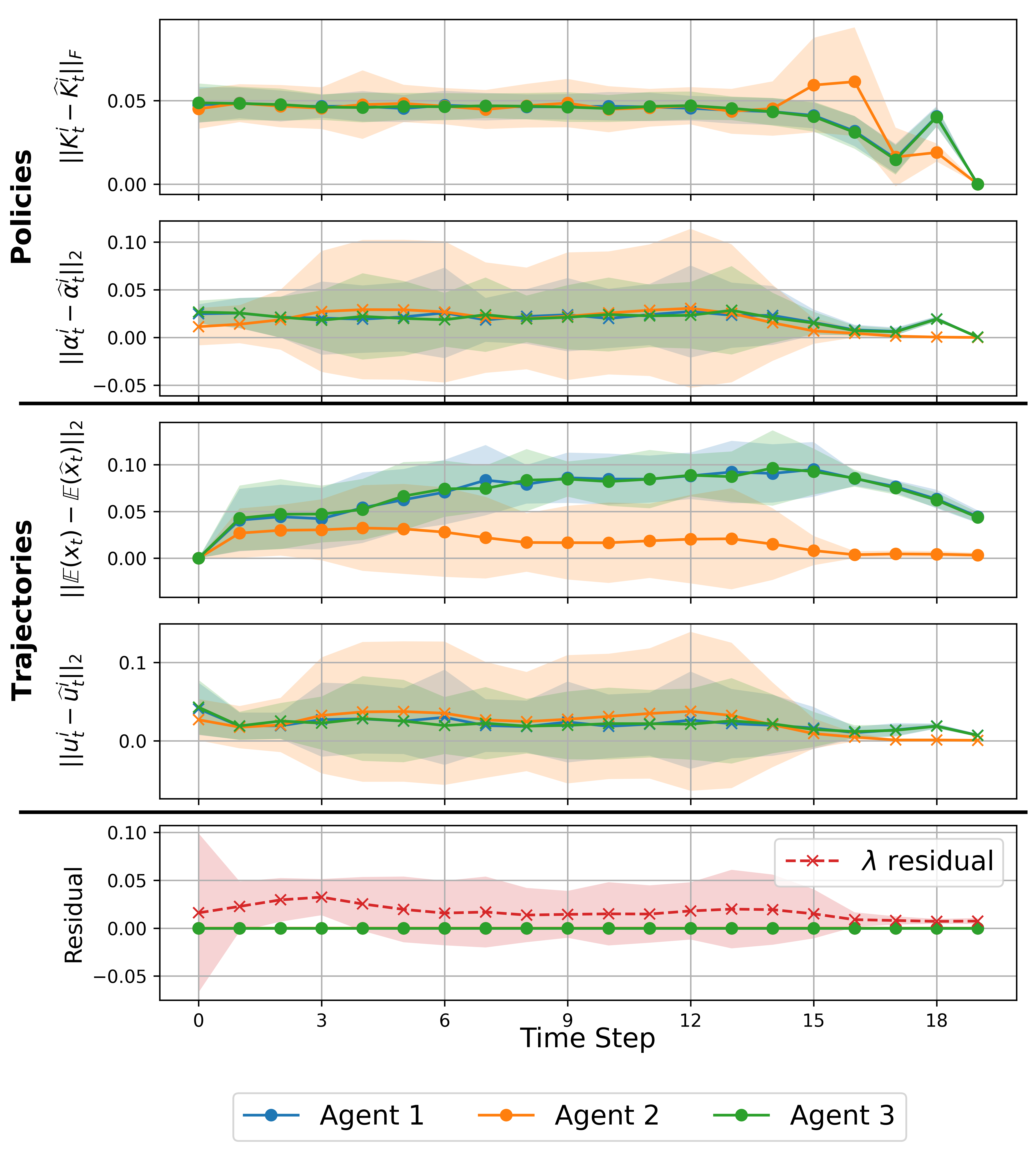} \vspace{-3mm}
    \caption{\textit{Optimization residuals, and deviations in the policy, state, and input trajectories generated from the identified cost parameters and optimal dual values relative to the expert parameters, across 100 randomized cost matrices.}} \vspace{-1mm}
    \label{fig:norm-diff-constrained-scalar}
\end{figure}

\section{Case studies} \label{sec:caseStudy}  \vspace{-1mm}
In this section, we aim to a) test, in a constrained setting, whether Algorithm~\ref{alg:constrainedILQG} can identify cost parameters that reproduce a given generalized Nash equilibrium policy and trajectories, and b) empirically examine the transferability result of Theorem~\ref{theorem:transferability}. To address a), we test our method in a numerical example (Section~\ref{sec:numerical}). To address b), we consider a multi-vehicle driving scenario (Section~\ref{sec:multicar}) and a real-world robot experiment (Section~\ref{sec:robot}). We empirically evaluate the perturbations in the optimal cost values induced by the identified cost parameters under the perturbed dynamics. We then compare these errors with the predictions from Theorem~\ref{theorem:transferability}.

\subsection{Cost identification in constrained LQG game} \label{sec:numerical}

We conducted a numerical simulation of an $N$-player constrained LQG game \eqref{problem:primal}. The game was defined with $N=3$ players and states $x_t \in \bbR^3$. The initial state $x_0$ was drawn from a Gaussian distribution \eqref{eq:Gaussian} with $\mu_0 = [1.0, \; 0.0, \; -1.0]^\top$ and $\Sigma_0$ being a $3\times 3$ identity matrix. The input for each player was $u^i_t \in \bbR$, and the time horizon was $T=20$. For all $t \in [T]^-$, we chose dynamic matrices $A_t$ to be a $3 \times 3$ identity matrix and $B^1_t  = [1, 0, 0]^\top$, $B^2_t  = [0, 1, 0]^\top$, $B^3_t  = [0, 0, 1]^\top$.  We generated 100 randomized LQG games, where in each game, the cost matrices $Q_t^i, R_t^i$ were diagonal with random positive entries. Collision avoidance in expected state trajectories was enforced as in \eqref{constraint:primalAffine}, and was shared and known among all agents: $\bbE[x^1_t - x^2_t] \geq 0.5, \;\; \bbE[x^2_t - x^3_t] \geq 0.5, \; \forall t \in [T]$. We chose simple cost and dynamics matrices to demonstrate the effectiveness of Algorithm~\ref{alg:constrainedILQG} and for better visualization, where all agents tried to get close to the origin while respecting the expectation state constraints. The method readily extends to more complex dynamics and cost matrices.

\begin{figure}[t]
    \centering
    \includegraphics[width=0.9\linewidth]{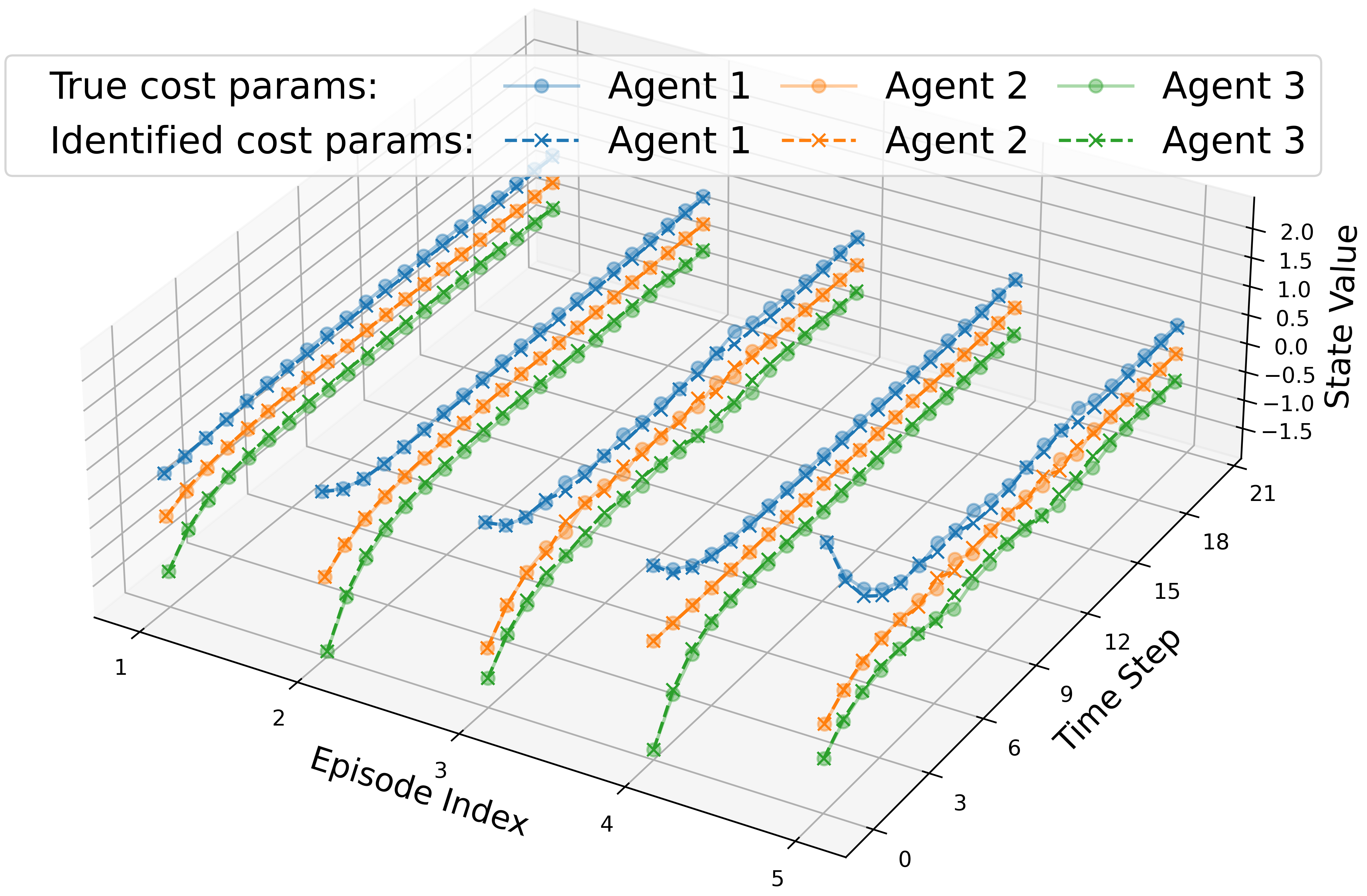}
    \caption{\textit{Five exemplary episodes, where we compare the trajectories generated by the identified cost parameters with the ground-truth. The trajectories generated by the identified cost parameters closely align with the ground-truth in all episodes.}}
    \label{fig:trajectories_num} 
\end{figure}

\begin{figure*}[ht]
    \centering
    \includegraphics[width=\textwidth]{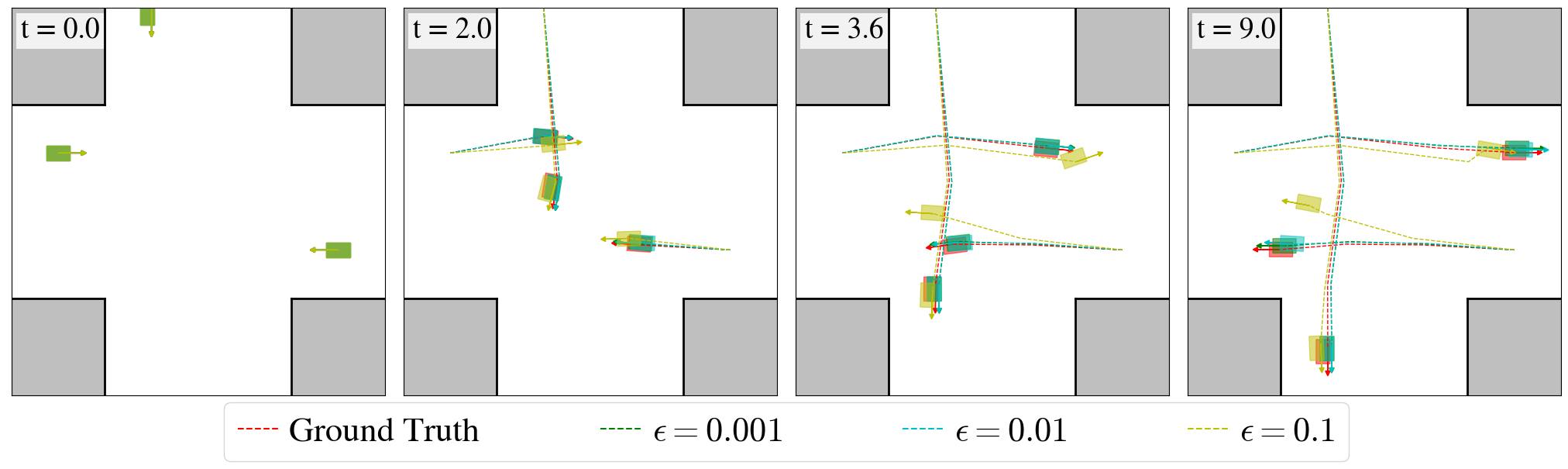}
    \caption{\textit{Trajectories from the expert cost parameters and the identified parameters. When the perturbations in the cost parameters and in the dynamics are small, the trajectories are close to the ground-truth. When the perturbation increases to $\epsilon = 0.1$, the recovered trajectory deviates significantly from the ground-truth.}}
    \label{fig:trajectory_params_noise}
\end{figure*}
\begin{figure}[t]
    \centering
    \includegraphics[width=0.95\linewidth]{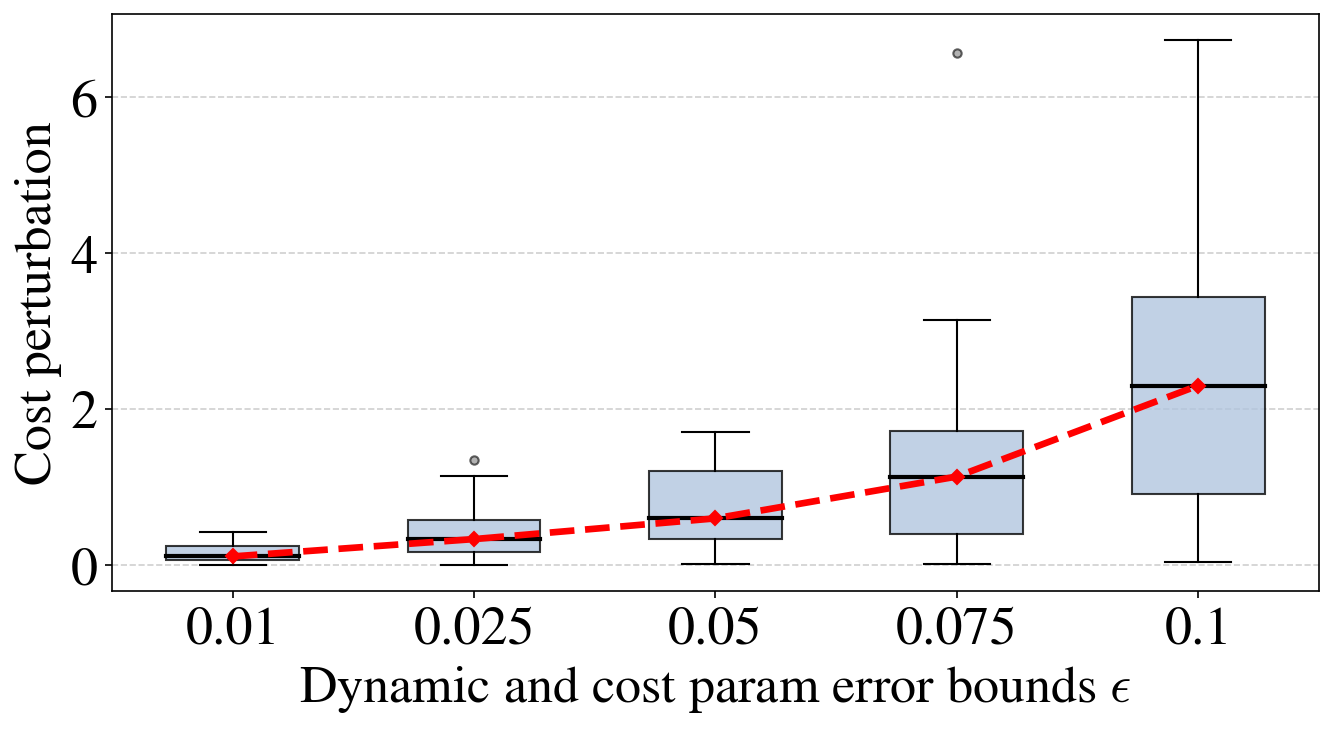}
    \caption{\textit{Cost perturbations as defined in \eqref{eq:accumCostPerturbation} for the multi-car example, where we measure the perturbation of the cumulative cost value over all agents with respect to cost parameters error and dynamics errors $\epsilon$. When $\epsilon \leq 0.075$, the median of the cost value perturbation scales linearly with $\epsilon$. With a larger $\epsilon$, the cost value perturbation increases more drastically.}}
    \label{fig:cost_error_boxplot}
\end{figure}

In each game, we first obtained the generalized Nash equilibrium policy $(K^{*}, \alpha^{*})$ \cite{Ren2025}. Then, we used Algorithm~\ref{alg:constrainedILQG} to identify the cost parameters from $(K^{*}, \alpha^{*})$. From the identified cost parameters, we recovered the policy $\{\widehat{K}^{i}_{t},\widehat{\alpha}^{i}_{t}\}^{i \in [N]}_{t \in [T]^-}$ and trajectories $\{\bbE\{\widehat{x}_{t+1}\},\widehat{u}^{i}_{t}\}^{i \in [N]}_{t \in [T]^-}$. The residual errors have two sources. First, solving~\eqref{eq:QP} and~\eqref{eq:LCPdual} incurs numerical error. Second, the generalized Nash equilibrium policy used as input to Algorithm~\ref{alg:constrainedILQG} is obtained via \cite[Algorithm 1]{Ren2025}, which itself is an approximate equilibrium. The identified cost parameters therefore inherit this approximation error. Figure~\ref{fig:norm-diff-constrained-scalar} illustrates the optimization residuals in~\eqref{eq:QP} and \eqref{eq:LCPdual}, and the norm differences between the recovered and the true policy, state, and input trajectories. The optimization residuals remain near zero, indicating that the identified parameters satisfy the characterizations \eqref{eq:characUnconstrained} and \eqref{eq:dualCharacterization}. Furthermore, the recovered policy parameters, states, and control inputs closely match the ground truth across 100 randomized episodes with varying cost matrices. 

Fig.~\ref{fig:trajectories_num} presents five examples of the recovered three-player trajectories obtained from the identified cost parameters, compared against the ground-truth. The recovered trajectories align closely with the true ones. The identified cost parameters by Algorithm~\ref{alg:constrainedILQG} successfully reproduce the true multi-agent behaviours and respect the collision avoidance constraints.

\subsection{Transferability of the identified cost parameters in an unconstrained LQG game} 
\label{sec:multicar}

In this section, we apply Algorithm~\ref{alg:constrainedILQG} to a driving scenario, and empirically examine the transferability of the identified cost parameters under perturbations on both the cost parameters and the dynamics.

We considered a three-vehicle interaction at an unsignalized cross intersection, illustrated in Fig.~\ref{fig:trajectory_params_noise}.
At $t=0$, three vehicles approached from different directions and aimed to cross as quickly as possible while avoiding collisions.
This interaction was modelled as an unconstrained LQG game, that is \eqref{problem:primal} without \eqref{eq:primalAffine}. The dynamics and cost parameters were defined as in~\cite[Section~IV.B]{ren2025identify}, and the collision avoidance is considered as a soft linear cost via local approximation.
We first computed the Nash equilibrium policies of all three agents.
Algorithm~\ref{alg:constrainedILQG}, with \texttt{setting} set to \texttt{unconstrained}, was then applied to identify the cost parameters from the Nash policies. The identified parameters ${\theta}^{id}$ were then used to recover the Nash equilibrium policies and the trajectories.

To assess transferability, we injected uniform noise $\epsilon$ into the identified cost parameters ${\theta}^{id}$ and the dynamics $(A,B)$, which gives $\widehat{\theta}^{id}$ and $(A',B')$, respectively.
For each noise level, $\epsilon$ was sampled uniformly within $\pm 10^{-4}$ of the nominal value. Figure~\ref{fig:trajectory_params_noise} shows representative closed-loop trajectories at three error levels $\epsilon \in \{10^{-3}, 10^{-2}, 10^{-1}\}$. At $\epsilon = 10^{-1}$ the deviation is pronounced and the original intersection crossing trajectories are no longer recovered. \vspace{1mm}

Figure~\ref{fig:cost_error_boxplot} reports the cost value perturbation summed over all agents across $30$ independent noise realizations 
{\scriptsize
\begin{equation} \label{eq:accumCostPerturbation}
\sum_{i=1}^N \left[\mathcal{J}^i\!\left(\theta^E, (K^*, \alpha^*)_{(A',B',\theta^E)}\right) - \mathcal{J}^i\!\left(\theta^E, (K^*, \alpha^*)_{(A',B',\widehat{\theta}^{id})}\right)\right],
\end{equation}}
When the error on dynamics and cost parameters satisfies $\epsilon \leq 0.075$ , the median cost value perturbation grows linearly with $\epsilon$, which is consistent with the theoretical result in Theorem~\ref{theorem:transferability}. When $\epsilon > 0.075$, the median cost value perturbation grows more drastically, indicating that the identified cost parameters are no longer transferable to the perturbed dynamics. This is consistent with Theorem~\ref{theorem:transferability} and the trajectory deviations observed in Fig.~\ref{fig:trajectory_params_noise}.

\subsection{Robot experiments for transferability} \label{sec:robot}

To examine the real-world performance, we further validate the transferability result of Theorem~\ref{theorem:transferability} on a robot testbed, consisting of three \textit{Nvidia JetBots} running on the Robot Operating System (ROS2). We use an \textit{OptiTrack} motion capture system to measure the states of the robots.  \vspace{1mm}

Each Jetbot $i$ follows a unicycle model with state $(p_x^i,p_y^i,\phi^i)$. Its linear and angular velocities are \(v^i= \tfrac{R_w}{2} \big(k^{i,r}u^{i,r}+k^{i,l}u^{i,l}\big), \quad \omega^i=\tfrac{R_w}{L_a} \big(k^{i,r}u^{i,r}-k^{i,l}u^{i,l}\big), \)
where $R_w$ is the wheel radius and $L_a$ is the axle length, and $k^{i,l}, k^{i,r}$ are the left and right wheel gains. Linearizing around a nominal trajectory $(\bar p_{x,t}^i,\bar p_{y,t}^i,\bar\phi_t^i,\bar v_t^i)$ and discretizing with step \(\Delta t\) gives linear dynamics matrices 
{\scriptsize
\begin{equation*}
A_t^i \approx \Delta t\begin{bmatrix} 0 & 0 & -\bar v_t^i\sin\bar\phi_t^i \\ 0 & 0 & \bar v_t^i\cos\bar\phi_t^i \\ 0 & 0 & 0 \end{bmatrix}\!, \;\;
B_t^i \approx \Delta t\tfrac{R_w}{2}\begin{bmatrix} k^{i,l}\cos\bar\phi_t^i & k^{i,r}\cos\bar\phi_t^i \\ k^{i,l}\sin\bar\phi_t^i & k^{i,r}\sin\bar\phi_t^i \\ -\tfrac{2}{L_a}k^{i,l} & \tfrac{2}{L_a}k^{i,r} \end{bmatrix}.
\end{equation*}}

On different surfaces, the change of wheel gains $k^{i,l}, k^{i,r}$ changes $B_t^i$, while $A_t$ remains unchanged. By conducting an identification on the wheel gains, we can estimate the dynamics parameters for each surface. 

\begin{figure}[t]
    \centering
    \begin{subfigure}[b]{0.49\linewidth}
        \centering
        \includegraphics[width=\linewidth]{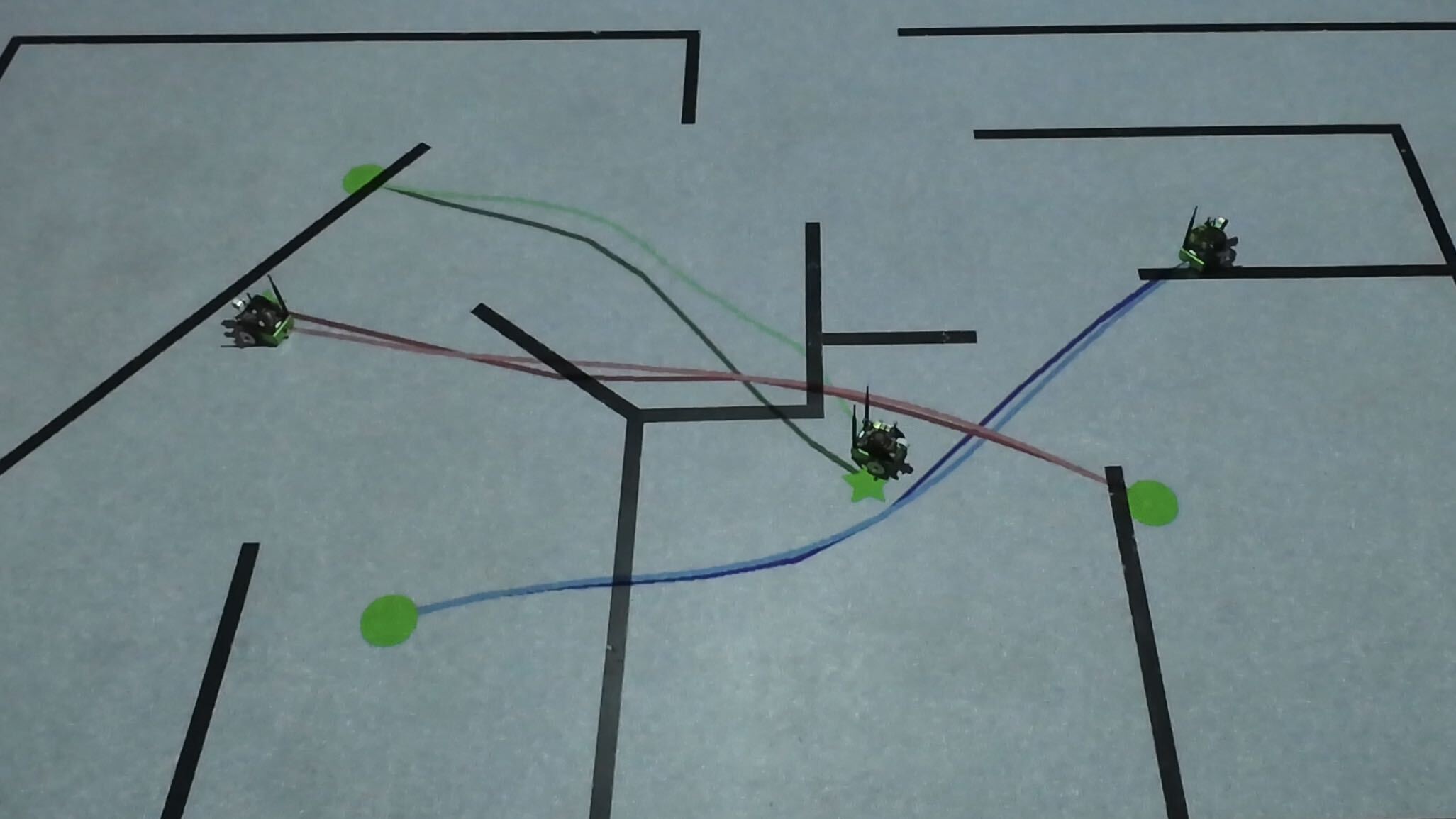}
        \caption{Lab floor (expert)}
    \end{subfigure}
    \hfill
    \begin{subfigure}[b]{0.49\linewidth}
        \centering
        \includegraphics[width=\linewidth]{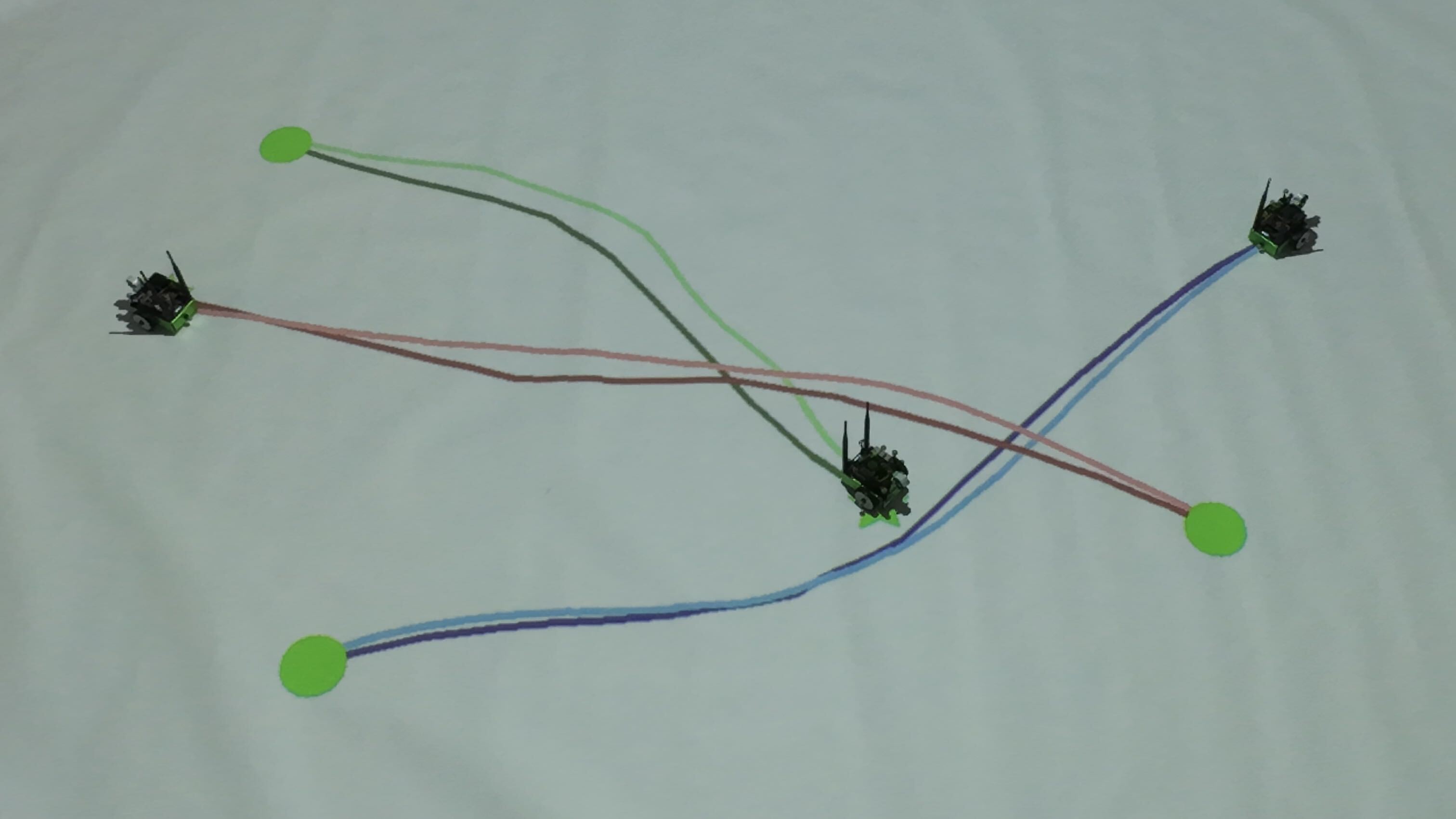}
        \caption{Paper}
    \end{subfigure}
    \\[1ex]
    \begin{subfigure}[b]{0.49\linewidth}
        \centering
        \includegraphics[width=\linewidth]{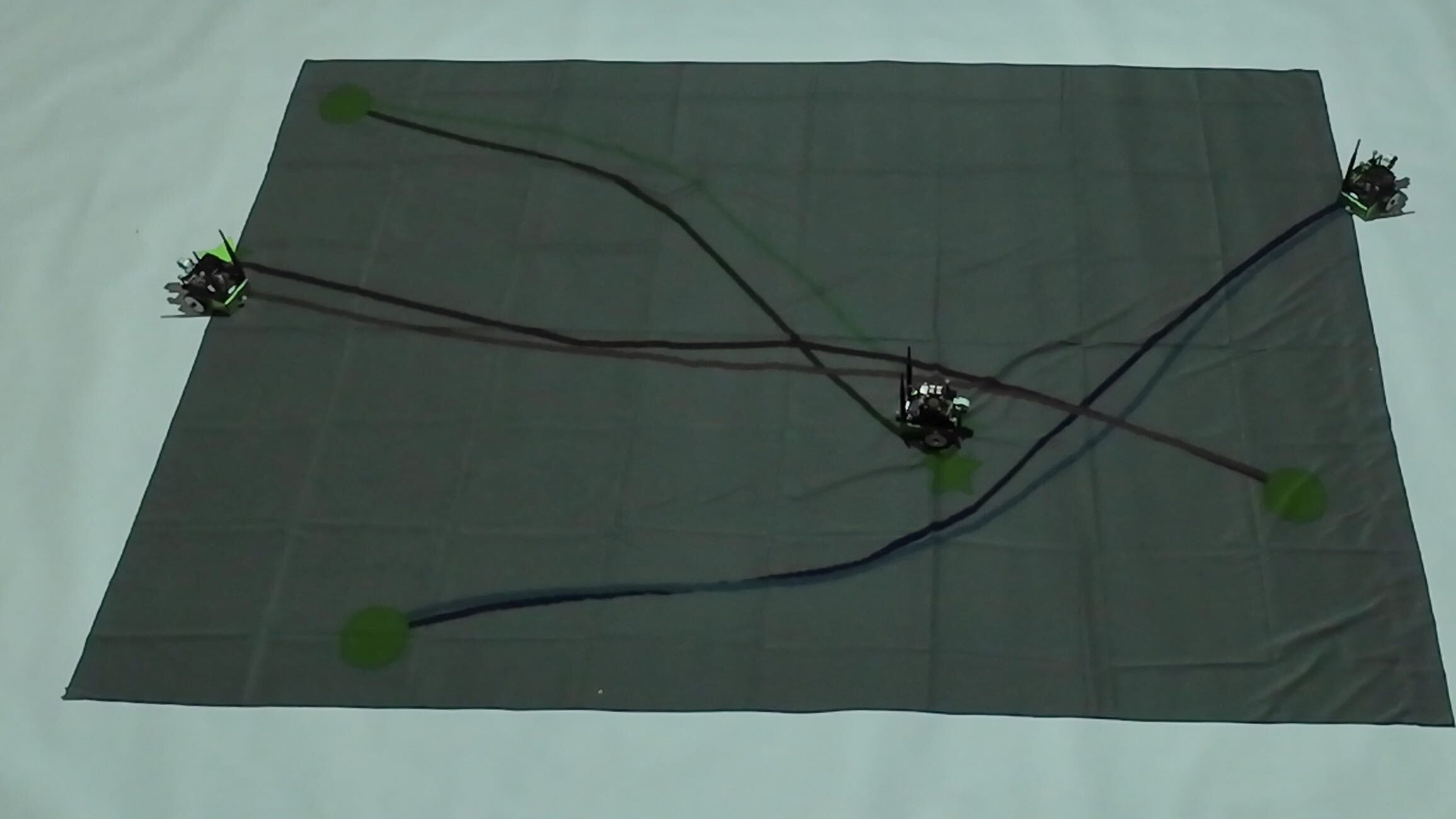}
        \caption{Fabric}
    \end{subfigure}
    \hfill
    \begin{subfigure}[b]{0.49\linewidth}
        \centering
        \includegraphics[width=\linewidth]{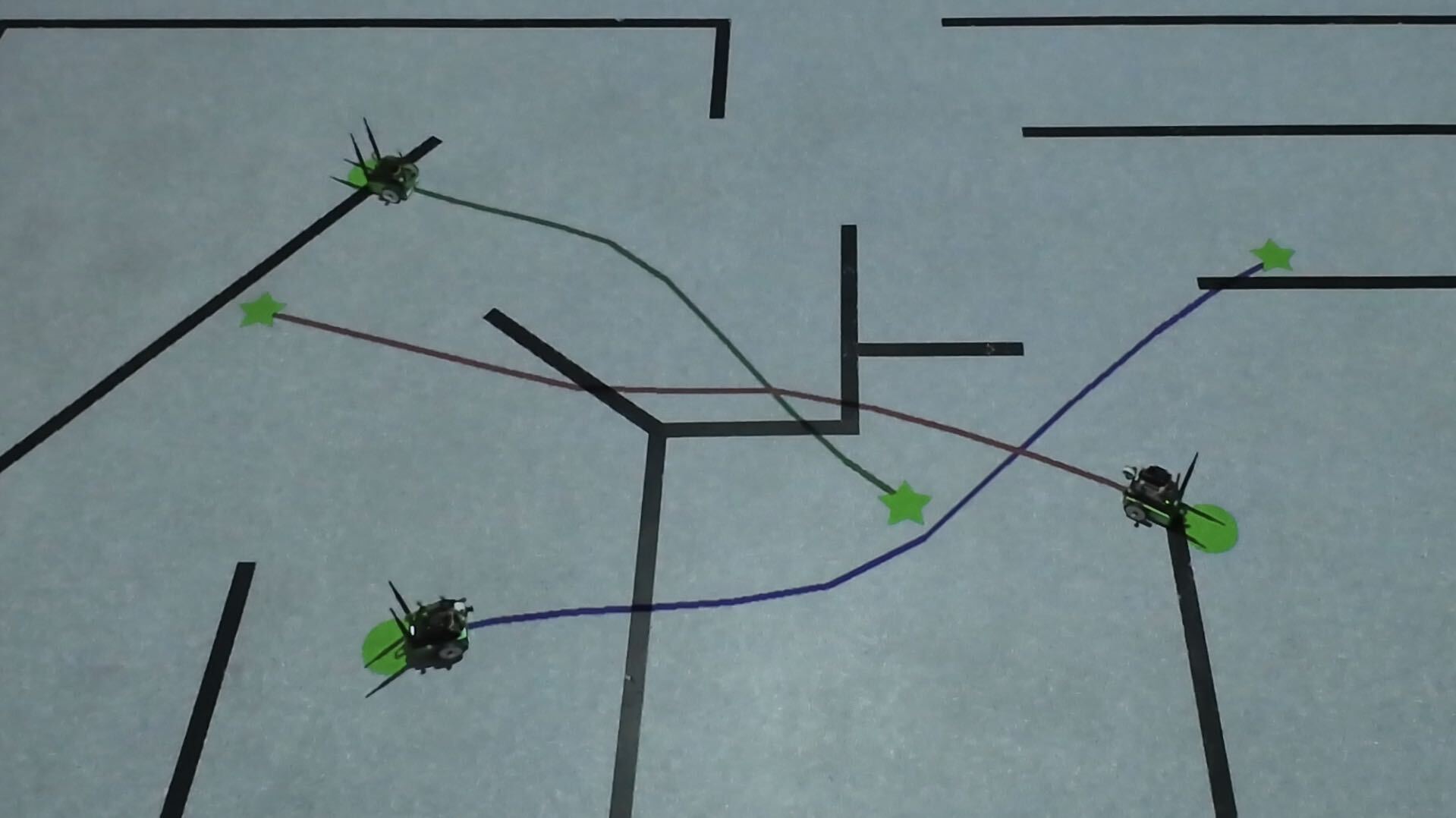}
        \caption{Wet floor}
    \end{subfigure} 
    \caption{\textit{Example trajectories of the three JetBots on four surfaces. For each robot, the dark-colored (blue, green, or red) line shows the expected trajectory under the expert cost parameters and the surface dynamics. The light-colored line shows the trajectory actually executed on that surface using the identified cost parameters. The identified cost parameters approximately reproduce the expert trajectories on the lab floor, paper, and fabric. On the wet floor, the robots' wheels slipped and failed to move (no light-colored line).}}
    \label{fig:robot-surfaces}
\end{figure}

We formulated an unconstrained LQG game (i.e., problem \eqref{problem:primal} without \eqref{eq:primalAffine}, same as section~\ref{sec:multicar}) where each robot aims to reach its own destination. 
We first computed the Nash equilibrium $(K^{*},\alpha^{*})$ on the expert dynamics $(A^E,B^E)$ corresponding to the lab floor. We then applied Algorithm~\ref{alg:constrainedILQG} with an unconstrained setting to identify the cost parameters $\widehat \theta^{id}$ from $(K^{*},\alpha^{*})$. We computed the Nash policy with the identified cost parameters on four surfaces to control the robots: the original lab floor, a sheet of paper, fabric, and a wet floor.

Fig.~\ref{fig:robot-surfaces} shows, for each surface, actual trajectories of the three robots (light blue, green, and red) controlled by the policy induced by the identified cost parameters and the perturbed dynamics. The expected trajectories under the expert cost parameters and the corresponding surface dynamics are shown in dark blue, green, and red, respectively. The trajectories on the lab floor, paper and fabrics approximately reproduce the trajectories under expert's cost parameters with some deviations. On the wet floor, the robots' wheels slipped and failed to move.

Table~\ref{tab:robot-surfaces} reports, for each surface, the dynamics error bound $\epsilon$ as defined in \eqref{eq:dynamicNeighbour}, and the average cost value perturbation as defined in \eqref{eq:accumCostPerturbation} over 10 runs. Except for the wet floor, the cost value perturbation scales roughly linearly with $\epsilon$, consistent with Theorem~\ref{theorem:transferability}. The identified cost parameters are not transferable to the wet floor, as the dynamics' error is too large.

\begin{table}[t]
\centering
\caption{\textit{Dynamics perturbation $\epsilon$ as defined in \eqref{eq:dynamicNeighbour},  and average cost value perturbation (mean over 10 runs) on each surface, relative to the expert dynamics and cost parameters.}} \vspace{-1mm}
\begin{tabular}{lcc}
\hline
\textbf{Surface} & $\epsilon$ & Avg. cost value perturbation  \\
\hline
\textbf{Lab floor (expert)} & 0.0 & 0.77 \\
\textbf{Paper} & 0.195 & 2.46  \\
\textbf{Fabric} & 0.618 & 6.13 \\
\textbf{Wet floor} & 2.81 & 42.25 \\
\hline
\end{tabular} \vspace{-3mm}
\label{tab:robot-surfaces}
\end{table}

\section{Conclusion}
We studied finite-horizon inverse linear quadratic Gaussian game in this work. Given a generalized Nash equilibrium policy in a constrained setting, the cost parameters and optimal dual values can be identified through tractable quadratic programs. In an unconstrained setting, under small errors in the identified cost parameters and in a sufficiently close neighbourhood of the expert dynamics, the Nash policies and cost values remain close to those under the expert cost parameters. Through numerical simulations, we showed that the proposed method recovers the trajectories corresponding to the generalized Nash equilibrium in a constrained setting. Real-robot experiments further demonstrated the transferability of the identified cost parameters to sufficiently close dynamics in the unconstrained setting. Overall, these results indicate that the cost parameters governing multi-agent interactions can be recovered from observed equilibrium policy or finite demonstrations. The identified cost parameters can be used to predict or design policies for similar systems, with the performance degrading linearly with the deviation of the dynamics and cost parameters. Future work will investigate the identifiability of cost parameters and dual values, and extend the transferability analysis to constrained dynamic games.

\bibliographystyle{IEEEtran}
\bibliography{autosam}

\appendices

\section{Proof of Lemma~\ref{lemma:policyBound}} \label{appendix:lemma1Proof}
By rearranging and stacking all agents' feedback gain and bias of the linear quadratic Gaussian (LQG) game, we get
{\small
\begin{align} \label{eq:Keq}
    \setlength{\arraycolsep}{1pt}
    \begin{bmatrix}
        K^{1*}_t \\ \vdots \\ K^{N*}_t
    \end{bmatrix} &= \Phi_t^{-1} \begin{bmatrix} (B_{t}^{1})^\top & \ldots & 0\\
        \vdots &  \ddots & \vdots\\  0 & \ldots & (B_{t}^{N})^\top 
    \end{bmatrix} \begin{bmatrix}
        P^{1*}_{t+1}A_t \\ \vdots \\ P^{N*}_{t+1}A_t 
    \end{bmatrix}, \\
    \begin{bmatrix}
        \alpha^{1*}_t \\ \vdots \\ \alpha^{N*}_t
    \end{bmatrix} &= \Phi_t^{-1}\begin{bmatrix} (B_{t}^{1})^\top & \ldots & 0\\
        \vdots &  \ddots & \vdots\\  0 & \ldots & (B_{t}^{N})^\top 
    \end{bmatrix} \begin{bmatrix}
        \zeta^{1*}_{t+1} \\ \vdots \\ \zeta^{N*}_{t+1}  \label{eq:alphaeq}
    \end{bmatrix},
\end{align} }

\noindent
where \(\Phi_t\) is a block matrix defined in Appendix~\ref{appendix:params}. The Riccati parameters $P_t^{i*}$ and $\zeta_t^{i*}$ are defined in \eqref{recatti:P} and \eqref{recatti:zeta}, respectively. \vspace{-2mm}

\subsection{Preliminary parameters' bounds.} 

\textit{(Bounding $\|\Phi_t\|$ and $\|\Delta\Phi_t\|$ in terms of $\|\Delta P_{t+1}\|$).}
Each block of $\Phi_t$ is either
\begin{align*}
\Phi_t^{ii}=R_t^i+(B_t^i)^\top P_{t+1}^{i*}B_t^i,
\qquad
\Phi_t^{ij}=(B_t^i)^\top P_{t+1}^{i*}B_t^j.
\end{align*}
Thus, using $\|\Delta R_t^i\|\le\epsilon$ and $\|B_t^i\|\le\xi$, and define $\|\Delta P_{t+1}\|:=\max_{i\in[N]}\|\Delta P_{t+1}^i\|$ and $ \|\Delta\zeta_{t+1}\|:=\max_i\|\Delta\zeta_{t+1}^i\|$, we get
\begin{align*}
\|\Phi_t^{ii}\|
&\le
\|R_t^i\| + \|(B_t^i)^\top P_{t+1}^{i*}B_t^i\|
\le
\|R_t^i\| + \|B_t^i\|^2\,\|P_{t+1}^{i*}\|
\\&\le
\xi + \xi^2\,\|P_{t+1}\|,
\end{align*}
\begin{align*}
\|\Phi_t^{ij}\|
\le
\|(B_t^i)^\top P_{t+1}^{i*}B_t^j\|
\le
\|B_t^i\|\,\|P_{t+1}^{i*}\|\,\|B_t^j\|
\le
\xi^2\,\|P_{t+1}\|.
\end{align*}
Similarly,
\begin{align*}
\|\Delta\Phi_t^{ii}\|
\le
\|\Delta R_t^i\|+\|(B_t^i)^\top(\Delta P_{t+1}^i)B_t^i\|
\le
\epsilon+\xi^2\|\Delta P_{t+1}\|,
\end{align*}
\begin{align*}
\|\Delta\Phi_t^{ij}\|
=
\|(B_t^i)^\top(\Delta P_{t+1}^i)B_t^j\|
\le
\xi^2\|\Delta P_{t+1}\|.
\end{align*}
Summing up the blocks, we get
\begin{equation*}
\|\Phi_t\|
\le
\cO\!\Big(\sqrt{N} \xi + \sqrt{N}\xi^2\|P_{t+1}\|\Big).
\end{equation*}
\begin{equation*}
\|\Delta\Phi_t\|
=
\cO\!\Big(\sqrt{N}\epsilon + \sqrt{N}\xi^2\|\Delta P_{t+1}\|\Big).
\end{equation*}

\textit{(Bounding \(\Delta K_t\) and \(\Delta \alpha_t\) in terms of \(\|\Delta P_{t+1}\|\) and \(\|\Delta\zeta_{t+1}\|\)).}  Let us denote $D_t$ the stacked vector with each block being $(B_t^i)^\top P_{t+1}^{i*}A_t$. Hence
\begin{align*}
\|D_t\|
=
\cO\!\big(\sqrt{N}\|B_t\|\|A_t\|\|P_{t+1}\|\big)
=
\cO\!\big(\sqrt{N}\xi^2\|P_{t+1}\|\big).
\end{align*}
\begin{align*}
\|\Delta D_t\|
=
\cO\!\big(\sqrt{N}\|B_t\|\|A_t\|\|\Delta P_{t+1}\|\big)
=
\cO\!\big(\sqrt{N}\xi^2\|\Delta P_{t+1}\|\big).
\end{align*}

Using \(\Delta(\Phi^{-1}) = -\Phi^{-1}(\Delta\Phi)\hat\Phi^{-1}\), we get \(\Delta(\Phi^{-1}) = \Delta\Phi = \cO(\epsilon)\) when $\epsilon$ is arbitrarily small \cite[Lemma 4]{ren2025identify}, we get
\begin{align} 
\|\Delta K_t\|
& \;\le\;
\|\Delta(\Phi_t^{-1})\|\,\|D_t\| + \|\hat\Phi_t^{-1}\|\,\|\Delta D_t\| \notag\\
& = \cO\!\Big(
\sqrt{N}\xi^5\,\epsilon + \sqrt{N}\xi^3\,\|\Delta P_{t+1}\|
\Big). \label{eq:DeltaK-dominant}
\end{align}

Similarly for $\alpha_t$, we get
\begin{equation}\label{eq:Da-dominant}
\|\Delta\alpha_t\|
=
\cO\!\Big(
\sqrt{N}\xi^4\,\epsilon
+
\sqrt{N}\xi^2\,\|\Delta P_{t+1}\|
+
\xi^2\,\|\Delta\zeta_{t+1}\|
\Big),
\end{equation}

Since $F_t = A_t - \sum_{j=1}^N B_t^j K_t^{j*}$,
\begin{align}\label{eq:F-bound}
\|F_t\|
&\le
\|A_t\| + {\big\|[B_t^1\ \cdots\ B_t^N]\big\|\,\|K_t^*\|}
=
\cO(\sqrt{N}\xi^2).
\end{align}
\begin{align*}
\|\Delta F_t\|
&\le
{\big\|[B_t^1\ \cdots\ B_t^N]\big\|\,\|\Delta K_t\|}
=
\cO\!\Big(
\sqrt{N}\xi\,\|\Delta K_t\|
\Big) \\
& = \cO\!\Big(
N\left(\xi^6\,\epsilon + \xi^4\,\|\Delta P_{t+1}\|\right)
\Big).
\end{align*}

\subsection{Recursion on bounding \texorpdfstring{$\|P_t^i\|$}{} and \texorpdfstring{$\|\Delta P_t^i\|$}{}}

\textit{(Base case $t=T$).} From $P_T^{i*}=Q_T^i$, we have $\|P_T^i\| = \|Q_T^i\|\le \xi$, \(
\|\Delta P_T^i\| = \|\Delta Q_T^i\|\le \epsilon.
\) 

\textit{(Recursion on $\|P_t^i\|$).} From \eqref{recatti:P}, 
\begin{align*}
\|P_t^i\|
&\le
\|F_t\|^2 \|P_{t+1}^i\|
+
\|K_t^{i*}\|^2 \|R_t^i\|
+
\|Q_t^i\| \\
&\le
\cO(N\xi^4)\|P_{t+1}^i\|+\cO(\xi^3).
\end{align*}
The second inequality holds due to \eqref{eq:uniformNormB} and \eqref{eq:F-bound}. Unrolling from $T$ with $\|P_T^i\|\le\xi$ gives
\begin{equation} \label{eq:Pt-norm-bound}
\|P_t^i\|=\cO\!\big(\xi (N\xi^4)^{T-t}\big).
\end{equation}

\textit{(Recursion on $\|\Delta P_t^i\|$).}
Taking differences in \eqref{recatti:P},
{\small \begin{align*}
\|\Delta P_t^i\|
&\le \cO(\epsilon) +
\|F_t\|^2 \|\Delta P_{t+1}^i\|
+
\cO\!\big(\|F_t\|\,\|P_{t+1}^i\|\,\|\Delta F_t\|\big) \\
& \quad +
\cO\!\big(\|K_t^{i*}\|\,\|R_t^i\|\,\|\Delta K_t^i\|\big)
+
\cO\!\big(\|K_t^{i*}\|^2\|\Delta R_t^i\|\big).
\end{align*}}
Using \eqref{eq:DeltaK-dominant}-\eqref{eq:Pt-norm-bound}, we obtain
{\small\begin{align*}
  \|\Delta P_t^i\|
&\le
\cO\!\big(N\xi^4\|\Delta P_{t+1}\|\big)
+
\cO\!\big(N^{3/2}\xi^6 \|P_{t+1}\|\|\Delta P_{t+1}\|\big)\\
& \qquad
+
\cO\!\big(N^{3/2}\xi^8 \|P_{t+1}\|\epsilon\big).  \\
& = \cO\!\Big(
N^{3/2}\xi^6 \|P_{t+1}\|\,\|\Delta P_{t+1}\|
+
N^{3/2}\xi^8 \|P_{t+1}\|\epsilon
\Big).
\end{align*}}
Substituting \eqref{eq:Pt-norm-bound} yields
\begin{align*}
\|\Delta P_t^i\|
=
\cO\!\Big(&
N^{3/2}\xi^7 (N\xi^4)^{T-t-1}\|\Delta P_{t+1}\| \\&
+
N^{3/2}\xi^{9}(N\xi^4)^{T-t-1}\epsilon
\Big).
\end{align*}
Unrolling from $T$ with $\|\Delta P_T^i\|\le\epsilon$ gives
\begin{equation*}
\|\Delta P_t\|
=
\cO\!\Big(
\epsilon\,
(N^{3/2}\xi^9)^{T-t}\,
(N\xi^4)^{\frac{(T-t-1)(T-t)}{2}}
\Big).
\end{equation*}

\subsection{Recursion on \texorpdfstring{$\|\zeta_t^i\|$}{}  and \texorpdfstring{$\|\Delta\zeta_t^i\|$}{}.}\vspace{4mm}

\noindent
\textit{(Base case $t=T$).}
From $\zeta_T^{i*}=l_T^i$, \(\|\Delta \zeta_T^i\|=\|\Delta l_T^i\|\le \epsilon.\)
\vspace{4mm}

\noindent
\textit{(Recursion on $\|\zeta_t^i\|$).}
From \eqref{recatti:zeta}, and taking norms of the summation term $\big\|\sum_{j=1}^N B_t^j\alpha_t^{j*}\big\|\le \sqrt{N}\|B_t\|\,\|\alpha_t\|$,
\begin{align*}
\|\zeta_t^i\|
&\le
\|F_t\|\,\|\zeta_{t+1}^i\|
+
\|F_t\|\,\|P_{t+1}^i\|\,\sqrt{N}\|B_t\|\,\|\alpha_t\|\\
& \qquad+
\|K_t^i\|\,\|R_t^i\|\,\|\alpha_t^i\|
+
\|l_t^i\|.
\end{align*}
Using \eqref{eq:uniformNormB}, \eqref{eq:F-bound} and \eqref{eq:Pt-norm-bound}, we obtain
\begin{align*}
\|\zeta_t\|
&\le
\cO(\sqrt{N}\xi^2)\,\|\zeta_{t+1}\|
+
\cO\!\big(N\xi^4\|P_{t+1}\|\big)
+
\cO(\xi^3) \\
&=
\cO(\sqrt{N}\xi^2)\,\|\zeta_{t+1}\|
+
\cO\!\big(\xi^5 (N\xi^4)^{T-t-1}\big),
\end{align*}
Unrolling from $\|\zeta_T\|=\|l_T\|\le\xi$ yields
\begin{equation}
\label{eq:zetat-norm-bound}
\|\zeta_t\|
=
\cO\!\Big(
\xi\,(N\xi^4)^{T-t}
\Big).
\end{equation}

\textit{(Recursion on $\|\Delta\zeta_t^i\|$).}
Take differences in \eqref{recatti:zeta} and bound term-by-term:
{\small\begin{align*}
\|\Delta\zeta_t^i\|
&\le
\|F_t\|\,\|\Delta\zeta_{t+1}^i\|
+
\|\Delta F_t\|\,\|\zeta_{t+1}\|
\\
&\quad+
\|F_t\|\,\|P_{t+1}\|\,\sqrt{N}\|B_t\|\,\|\Delta\alpha_t\| \\
&\quad
+
\|F_t\|\,\sqrt{N}\|B_t\|\,\|\alpha_t\|\,\|\Delta P_{t+1}\|\\
&\quad
+
\|\Delta F_t\|\,\|P_{t+1}\|\,\sqrt{N}\|B_t\|\,\|\alpha_t\|
\\
&\quad+
\|K_t\|\,\|R_t\|\,\|\Delta\alpha_t\|
+
\|K_t\|\,\|\Delta R_t\|\,\|\alpha_t\|\\
&\quad
+
\|\Delta K_t\|\,\|R_t\|\,\|\alpha_t\|
+
\|\Delta l_t\|.
\end{align*}}
Using \eqref{eq:costParamBound},\eqref{eq:uniformNormB}, \eqref{eq:DeltaK-dominant}-\eqref{eq:zetat-norm-bound}, we get
\begin{align*}
\|\Delta\zeta_t\|
=
\cO\!\Big(
&N\xi^6\|P_{t+1}\|\,\|\Delta\zeta_{t+1}\|
+
N\xi^4\|P_{t+1}\|\,\|\Delta P_{t+1}\|
\\&+
N\xi^6\|P_{t+1}\|\,\epsilon
+
N\xi^4\,\|\zeta_{t+1}\|\,\|\Delta P_{t+1}\|\\&
+
N\xi^6\,\|\zeta_{t+1}\|\,\epsilon
\Big).
\end{align*}
Substituting \eqref{eq:Pt-norm-bound}-\eqref{eq:zetat-norm-bound}, we get
\begin{align*}
\|\Delta\zeta_t\|
=
\cO\!\Big(&
N\xi^7 (N\xi^4)^{T-t-1}\,\|\Delta\zeta_{t+1}\|
+
N\xi^7 (N\xi^4)^{T-t-1}\,\epsilon \\& +
N\xi^5 (N\xi^4)^{T-t-1}\,\|\Delta P_{t+1}\|
\Big).
\end{align*}
As $\|\Delta P_{t+1}\|=\cO(\epsilon)$ along the backward recursion, hence we get
\begin{align*}
\|\Delta\zeta_t\|
=
\cO\!\Big(
N\xi^7 (N\xi^4)^{T-t-1}\,\|\Delta\zeta_{t+1}\|\Big).
\end{align*}
Unrolling from $T$ with $\|\Delta\zeta_T\|\le\epsilon$ gives
\begin{equation}
\|\Delta \zeta_t\|
=
\cO\!\Big(
\epsilon\,
(N\xi^7)^{T-t}\,
(N\xi^4)^{\frac{(T-t-1)(T-t)}{2}}
\Big).
\end{equation}

\noindent
\textit{(Bounding $\|\Delta K_t\|$ and $\|\Delta\alpha_t\|$).}
Using \eqref{eq:DeltaK-dominant},
\begin{align*}
\|\Delta K_t\|
&=\cO\!\big(\sqrt{N}\xi^3\,\|\Delta P_{t+1}\|\big)\\
&=\cO\!\Big(
\sqrt{N}\xi^3\,
\,
(N^{3/2}\xi^9)^{T-t-1}\,
(N\xi^4)^{\frac{(T-t-2)(T-t-1)}{2}} \epsilon
\Big).
\end{align*}
Using \eqref{eq:Da-dominant}, we get
\begin{align*}
\|\Delta\alpha_t\|
&=
\cO\!\big(\sqrt{N}\xi^2\,\|\Delta P_{t+1}\|+\xi^2\|\Delta\zeta_{t+1}\|\big)\\& =
\cO\!\big(\sqrt{N}\xi^2\,
(N\xi^9)^{T-t-1}\,
(N\xi^4)^{\frac{(T-t-2)(T-t-1)}{2}} \epsilon\big).
\end{align*}
Conservatively simplifying their exponents gives
\[
\|\Delta K_t\|
=
\cO\!\left(
N^{(T-t)^2}
\xi^{2(T-t+1)^2}
\epsilon
\right),
\]
\[
\|\Delta\alpha_t\|
=
\cO\!\left(
N^{\frac{(T-t)^2}{2}}
\xi^{2(T-t+1)^2}
\epsilon
\right).
\]
This yields the two bounds stated in Lemma~\ref{lemma:policyBound}. \(\hfill\square\)

\section{Proof of Lemma~\ref{lemma:policyBoundDyn}}
\label{appendix:lemma2Proof}

 We compare two games with the same cost parameter \(\theta\), but with the true dynamics \((A^E,B^E)\) and a neighbouring dynamics \((A',B')\). Let \((K_t^{i,E},\alpha_t^{i,E},P_t^{i,E},\zeta_t^{i,E})\) denote the Riccati quantities and Nash policy under \((A^E,B^E)\), and let \((K_t^{i\prime},\alpha_t^{i\prime},P_t^{i\prime},\zeta_t^{i\prime})\) denote the ones under \((A',B')\). In this section, we use \(\Delta(.)\) to denote the difference of parameters induced by \((A^E,B^E)\) and \((A',B')\).

Similar to the proof of Lemma~\ref{lemma:policyBound}, we first upper-bound the perturbation of $\|\Phi_t\|$ from the dynamics' error in terms of \(\Delta P_{t+1}\). 
Using \(\|\Delta B_t^i\|\le\epsilon\), we have
\begin{equation}\label{eq:DPhi-dyn-only}
\|\Delta\Phi_t\|=
\cO\!\left(
\sqrt N\,\xi\,\|P_{t+1}^{E}\|\epsilon
+\sqrt N\,\xi^2\|\Delta P_{t+1}\|
\right).
\end{equation}
Let \(D_t\) denote the stacked right-hand side in \eqref{eq:Keq}, whose \(i\)-th block is \((B_t^{i,E})^\top P_{t+1}^{i,E}A_t^E\) under the true dynamics and \((B_t^{i\prime})^\top P_{t+1}^{i\prime}A'_t\) under the perturbed dynamics. Using product-difference expansion, we have
\begin{equation}\label{eq:DD-dyn-only}
\|\Delta D_t\|=
\cO\!\left(
\sqrt N\,\xi\,\|P_{t+1}^{E}\|\epsilon
+\sqrt N\,\xi^2\|\Delta P_{t+1}\|
\right),
\end{equation}
Since \(\Delta(\Phi_t^{-1})=-\Phi_t^{-1}(\Delta\Phi_t)(\Phi_t')^{-1}\), and from \eqref{eq:DPhi-dyn-only}--\eqref{eq:DD-dyn-only}, we have
\begin{equation}\label{eq:DK-dyn-only}
\|K_t'-K_t^E\|
=\cO\!\left(
\sqrt N\,\xi^3\|P_{t+1}^{E}\|\epsilon
+\sqrt N\,\xi^3\|\Delta P_{t+1}\|
\right).
\end{equation}
Similarly, from \eqref{eq:alphaeq},
\begin{align}
\|\alpha_t'-\alpha_t^E\|
&=\cO\!\Big(
\sqrt N\,\xi^2
\big(\|P_{t+1}^{E}\|
+\|\zeta_{t+1}^{E}\|\big)\epsilon \notag\\
&\quad
+\sqrt N\,\xi^2\|\Delta P_{t+1}\|
+\xi^2\|\Delta\zeta_{t+1}\|
\Big).
\label{eq:Dalpha-dyn-only}
\end{align}
We now propagate the Riccati perturbations backward. Since the terminal costs are the same,
\(
\Delta P_T^i=0,\;\;
\Delta\zeta_T^i=0.
\) The nominal bounds used in Lemma~\ref{lemma:policyBound} remain valid here:
\[
\|P_t^{E}\|
=
\cO\!\big(\xi (N\xi^4)^{T-t}\big),
\qquad
\|\zeta_t^{E}\|
=
\cO\!\big(\xi (N\xi^4)^{T-t}\big).
\]
For the closed-loop dynamics, we have
\begin{align*}
\Delta F_t
&=\Delta A_t
-\sum_{j=1}^N \Delta B_t^jK_t^{j\prime}
-\sum_{j=1}^N B_t^{j,E}\!\big(K_t^{j\prime}-K_t^{j,E}\big),
\end{align*}
and therefore, by \eqref{eq:DK-dyn-only},
\begin{equation}\label{eq:DF-dyn-only}
\|\Delta F_t\|=
\cO\!\left(
N^{T-t+\frac12}\xi^{4(T-t)+1}\epsilon
+N^{\frac32}\xi^4\|\Delta P_{t+1}\|
\right).
\end{equation}
Taking differences in \eqref{recatti:P} and applying \eqref{eq:DK-dyn-only} and \eqref{eq:DF-dyn-only} gives
\begin{align*}
\|\Delta P_t\|
&=
\cO\!\Big(
N^{2(T-t)}\xi^{8(T-t)}\epsilon
+N^{T-t+1}\xi^{4(T-t)+3}\|\Delta P_{t+1}\|
\Big).
\end{align*}
Since \(\Delta P_T=0\), finite-horizon backward induction yields 
\begin{equation}\label{eq:DP-dyn-only}
\|\Delta P_t\|
=
\cO\!\Big(
\epsilon\,
N^{\frac{(T-t+2)^2}{2}}\,
\xi^{2(T-t+2)^2}
\Big).
\end{equation}
Taking differences in \eqref{recatti:zeta} and applying \eqref{eq:DK-dyn-only} - \eqref{eq:DP-dyn-only} yields
\begin{align*}
\|\Delta\zeta_t\|
&=
\cO\!\Big(
N^{2(T-t)}\xi^{8(T-t)}\epsilon
+N^{T-t+1}\xi^{4(T-t)+3}\|\Delta P_{t+1}\|
\\
&\qquad\qquad
+N^{T-t+1}\xi^{4(T-t)+3}\|\Delta\zeta_{t+1}\|
\Big).
\end{align*}
Since \(\Delta\zeta_T=0\), another backward recursion gives
\begin{equation}\label{eq:Dzeta-dyn-only}
\|\Delta\zeta_t\|
=
\cO\!\Big(
\epsilon\,
N^{\frac{(T-t+2)^2}{2}}\,
\xi^{2(T-t+2)^2}
\Big).
\end{equation}
Substituting \eqref{eq:DP-dyn-only} and \eqref{eq:Dzeta-dyn-only} into \eqref{eq:DK-dyn-only} and \eqref{eq:Dalpha-dyn-only}, and overapplying the bounds by completing the square gives
\begin{align*}
\|K_t'-K_t^E\|
&=
\cO\!\Big(
N^{\frac{(T-t+2)^2}{2}}\,
\xi^{2(T-t+2)^2}\epsilon
\Big),\\
\|\alpha_t'-\alpha_t^E\|
&=
\cO\!\Big(
N^{\frac{(T-t+2)^2}{2}}\,
\xi^{2(T-t+2)^2}\epsilon
\Big),
\end{align*}
This yields the two bounds stated in Lemma~\ref{lemma:policyBoundDyn}. \(\hfill\square\)

\section{Derivation of the transferability cost bound} \label{sec:cost-bound-derivation}
In this appendix, write the stacked mean trajectories as \(\vx=F\mu_0+Z\bm{\alpha}\) and \(\vu^i=-K^i\vx-\bm{\alpha}^i\). Under \eqref{eq:uniformNormB}, we have
\begin{align}
\big|\widehat{\cJ}^i-\cJ^i\big|
&\le
\big| \hat \vx^\top {\cQ}^i \hat \vx - \vx^\top \cQ^i \vx \big|
+
\big| {\cL}^{i\top}\hat \vx - \cL^{i\top}\vx \big| \notag \\
&\qquad +
\big| (\hat \vu^i)^\top {\cR}^i \hat \vu^i - (\vu^i)^\top \cR^i \vu^i \big|
+\mathcal E_{\mathrm{trace}}^i.
\label{eq:cost-split}
\end{align}
where \(
\mathcal E_{\mathrm{trace}}^i :=\sum_{t=0}^{T-1}\Big[
|\operatorname{tr}\:(Q_{t+1}^i(\widehat\Sigma_{t+1}-\Sigma_{t+1}))| +|\operatorname{tr}\:(R_t^i(\widehat K_t^i\widehat\Sigma_t\widehat K_t^{i\top}-K_t^i\Sigma_tK_t^{i\top}))|\Big]\) accounts for the trace term associated with initial-state uncertainty and accumulated process noise.

\subsection{Bounding the stacked terms.} 

Let us define the worst-case policy perturbation bound over all time steps,
{\small
\begin{subequations} \label{eq:barDK}
\begin{align}
\bar\Delta K
& := \max_{t\in[T]^-}\|\Delta K_t\|
= \cO\!\Big(N^{(T+2)^2}\xi^{2(T+2)^2}\epsilon \Big), \\
\bar\Delta \alpha
&:= \max_{t\in[T]^-}\|\Delta \alpha_t\|
= \cO\!\Big(N^{\frac{(T+2)^2}{2}}\xi^{2(T+2)^2}\epsilon \Big). 
\end{align}
\end{subequations}
}

\noindent
Let \(\bar F_t:={\max_{t\in[T]^-}\{\|F_t\|\}}=\cO({\sqrt N}\xi^2)\). 
Also from \eqref{eq:uniformNormB}: $\|\cQ^i\|\le \xi$, $\|\cR^i\|\le \xi$, $\|\cL^i\|\le \sqrt{T}\,\xi$. Moreover,
\(
\|\bm{\alpha}\|\le {\sqrt T}\,\xi, \; \|\Delta \bm{\alpha}\|\le \sqrt{T}\,\bar\Delta\alpha,\;\|K^i\|,\le \xi,\; \|\Delta K^i\|\le \bar\Delta K.
\) Hence,
\begin{align*}
\bar\Delta F_t &:= \max_{t\in[T]^-}\|\Delta F_t\| {\le \sqrt{N}\xi\,\bar\Delta K}.
\end{align*}

Each block row of $F$ contains products of $F_t$'s, and each block of $Z$ contains products of $F_t$'s and $B_t$.
Using a conservative product-difference bound yields, for any $m\le T$, $\Big\|\prod_{\tau=0}^{m}\hat F_\tau-\prod_{\tau=0}^{m}F_\tau\Big\|
\le (m+1)\,\bar F_t^{m}\,\bar\Delta F_t.$

Therefore, summing across the $(T+1)$ stacked block rows and using $\bar F_t\le{2\sqrt N\xi^2}$, we obtain
\begin{align*}
\|\Delta F\| &\le \cO\!\Big(T^2\,\bar F_t^{T}\,\bar\Delta F_t\Big), \quad \|\Delta Z\| \le \cO\!\Big(T^3\sqrt N\,\xi\bar F_t^{T}\,\bar\Delta F_t\Big). 
\end{align*}
where for \(\Delta Z\), the product-difference contains a factor \(\|B_t\|\le\sqrt N\xi\) amd each block has one \(\Delta F_t\) contributes at most a \(\bar\Delta F_t\) error term. Also,
\begin{equation}
\|F\|\le \cO(T\,\bar F_t^{T}),\qquad \|Z\|\le \cO(T^2\sqrt N\,\bar F_t^{T}\,\xi). \label{eq:FZnorms}
\end{equation}

\subsection{Bounding \texorpdfstring{$\|\Delta\mathbf{x}\|$}{} and \texorpdfstring{$\|\Delta\mathbf{u}^i\|$}{}.} 

Using \(\Delta\vx=(\Delta F)\mu_0+(\Delta Z)\bm{\alpha}+ Z(\Delta\bm{\alpha})\), 
\begin{align*}
\|\Delta\vx\|
&\le \|\Delta F\|\,\|\mu_0\| + \|\Delta Z\|\,\|\bm{\alpha}\| + \| Z\|\,\|\Delta\bm{\alpha}\| \notag\\
&\le \cO\!\Big( 
T^2\bar F_t^{T}\bar\Delta F_t\,\|\mu_0\|
+ T^2\sqrt N\bar F_t^{T}\xi\cdot \sqrt{T}\bar\Delta\alpha \notag \\
& \qquad\qquad + T^3\sqrt N\xi\bar F_t^{T}\bar\Delta F_t\,{\sqrt T}\xi
\Big). 
\end{align*}
For stacked inputs across all time steps for player $i$, 
\begin{align*}
\|\Delta \vu^i\|
&\le \|\Delta K^i\|\,\|\vx\| + \|K^i\|\,\|\Delta\vx\| + \|\Delta\alpha^i\| \notag\\
&\le \bar\Delta K\,\|\vx\| + \xi\,\|\Delta\vx\| + \sqrt T\bar\Delta\alpha. 
\end{align*}
Moreover, by \eqref{eq:FZnorms} and \(\|\bm{\alpha}\|\le {\sqrt T}\xi\),
\begin{align*}
\|\vx\|
&\le \|F\|\,\|\mu_0\| + \|Z\|\,\|\bm{\alpha}\| \notag \\
&\le \cO\!\Big(T\bar F_t^{T}\xi + T^2\sqrt N\bar F_t^{T}\xi\cdot {\sqrt T}\xi\Big) \notag\\
&= \cO\!\big(T^{5/2}{\sqrt N}\,\bar F_t^{T}\xi^2\big),
\end{align*}
\begin{align}
\|\vu^i\|
&\le
\|K^i\|\,\|\vx\| + \|\bm{\alpha}^i\| \le \xi\,\|\vx\| + \sqrt{T}\,\xi \notag \\
& = \cO\!\big(T^{5/2}{\sqrt N}\,\bar F_t^{T}\xi^3 + \sqrt{T}\xi\big).
\label{eq:ui-bound}
\end{align}

{
\subsection{Bounding the covariance trace terms.}
The covariance recursions are
\(
\Sigma_{t+1}=F_t\Sigma_tF_t^\top+W_t.
\)
Hence,
\(\|\Sigma_t\|\le(t+1)\xi\bar F_t^{2t}\).
Writing \(\Delta\Sigma_t:=\widehat\Sigma_t-\Sigma_t\), we obtain
\begin{align*}
\|\Delta\Sigma_{t+1}\|
&\le\bar F_t^2\|\Delta\Sigma_t\|
+2\bar F_t\|\Sigma_t\|\bar\Delta F_t.
\end{align*}

Unrolling from $t=0$ where \(\|\Delta\Sigma_0\|=0\), we can get
\(
\|\Delta\Sigma_t\| \le t(t+1)\xi\bar F_t^{2t-1}\bar\Delta F_t.
\) Similarly, 
\(
\|\widehat K_t^i \widehat\Sigma_t \widehat K_t^{i\top}-K_t^i\Sigma_tK_t^{i\top}\|
\le\xi^2\|\Delta\Sigma_t\|+2\xi\|\Sigma_t\|\bar\Delta K.
\)
Substituting \(\bar F_t\le{2\sqrt N\xi^2}\), \(\bar\Delta F_t{\le\sqrt N\xi\bar\Delta K}\), and \eqref{eq:barDK}, yields
\begin{align}
\mathcal E_{\mathrm{trace}}^i
&=\cO\!\Big(T^3\bar F_t^{2T}\xi^4
(\bar\Delta F_t+\bar\Delta K)\Big) \notag\\
&=\cO\!\Big(T^3
{4^T N^{T^2+5T+\frac92}}\xi^{2T^2+12T+13}\epsilon\Big).
\label{eq:cov-trace-bound}
\end{align}
}

\subsection{Transferability bound in Theorem~\ref{theorem:transferability}.} 
Based on \eqref{eq:cost-split}, we can further bound the cost deviation terms by
\begin{align*}
\big|\hat\vx^\top \cQ^i \hat\vx - \vx^\top \cQ^i \vx\big|
&\le \|\cQ^i\|\,\|\hat\vx+\vx\|\,\|\Delta\vx\| \\
&\le \xi\,(2\|\vx\|+\|\Delta\vx\|)\,\|\Delta\vx\|.
\end{align*}
\begin{align*}
\big|(\hat\vu^i)^\top \cR^i \hat\vu^i - (\vu^i)^\top \cR^i \vu^i\big|
&\le \|\cR^i\|\,\|\hat\vu^i+\vu^i\|\,\|\Delta\vu^i\|\\
&\le \xi\,(2\|\vu^i\|+\|\Delta\vu^i\|)\,\|\Delta\vu^i\|.
\end{align*}
\begin{align*}
\big|{\cL}^{i\top}\hat \vx - \cL^{i\top}\vx\big|
\le \|\cL^i\|\,\|\Delta\vx\|
\le \sqrt{T}\,\xi\,\|\Delta\vx\|.
\end{align*}
Hence, the overall cost deviation bound can be written as
\begin{align*}
\big|\widehat{\cJ}^i-\cJ^i\big|
&\le
\xi\,(2\|\vx\|+\|\Delta\vx\|)\,\|\Delta\vx\|
+
\sqrt{T}\,\xi\,\|\Delta\vx\| \notag \\
&\quad+ 
\xi\,(2\|\vu^i\|+\|\Delta\vu^i\|)\,\|\Delta\vu^i\|
{+\mathcal E_{\mathrm{trace}}^i}.
\end{align*}
Given that $\|\Delta\vx\| = \cO(\epsilon)$ and $\|\Delta\vu^i\| = \cO(\epsilon)$, we can drop the quadratic $\epsilon$ terms. To bound the linear terms, we note that $\xi\|\vu^i\|\|\Delta\vu^i\|$ carries the higher powers of $\xi$ and $T$, and thus dominates the term $\xi\|\vx\|\|\Delta\vx\|$.  Substituting the bounds for $\|\vx\|$ and $\|\Delta\vx\|$ we get
\begin{align*}
\|\Delta\vu^i\| &\le \bar\Delta K\|\vx\| + \xi\|\Delta\vx\| + \sqrt T\bar\Delta\alpha \\
&= \cO\!\Big( {\sqrt N}\bar F_t^T\xi^3 (T^{5/2}\bar\Delta K + T^{7/2}\bar\Delta F_t + T^{5/2}\bar\Delta\alpha) \Big).
\end{align*}

Multiplying $\|\Delta\vu^i\|$ by $\xi\|\vu^i\| = \cO(T^{5/2}{\sqrt N}\bar F_t^T\xi^4)$, we get
\begin{align*}
\big|\widehat{\cJ}^i-\cJ^i\big| &\le \cO\!\Big( T^6 {N}\bar F_t^{2T} \xi^7 \big(\bar\Delta F_t + \bar\Delta K + \bar\Delta\alpha\big) \Big)
{+\mathcal E_{\mathrm{trace}}^i}.
\end{align*}

Finally, substituting \(\bar\Delta F_t {\le\sqrt N\xi\bar\Delta K}\) and noting that ${\sqrt N}\xi\bar\Delta K$ dominates both \(\bar\Delta K\) and \(\bar\Delta\alpha\) term, we get
\begin{equation*}
\big|\widehat{\cJ}^i-\cJ^i\big|
\le
\cO\!\Big(
T^6{N}\,\bar F_t^{2T}\xi^7
\big({\sqrt N}\xi\,\bar\Delta K + \bar\Delta\alpha\big)
\Big){+\mathcal E_{\mathrm{trace}}^i}.
\end{equation*}
Substituting \eqref{eq:barDK}--\eqref{eq:cov-trace-bound} and collecting the dominant terms, we get
\begin{align*}
\big|\widehat{\cJ}^i-\cJ^i\big|
&\le \cO\!\Big(T^6{ N^{T^2+5T+\frac{11}{2}}}\xi^{2T^2+12T+16}\epsilon\Big).
\end{align*}
This is bounded by
\(\cO\!\Big(\epsilon T^6(N\xi^2)^{\cO(T^2)}\Big),\)
as stated in Theorem~\ref{theorem:transferability}.

\section{Matrices definitions} \label{appendix:params}
In this section, we define the matrices that are used but omitted in the main text.
{\footnotesize
\begin{align*}
    &\Phi_t = \begin{bmatrix} 
    R_{t}^{1} + (B_{t}^{1})^\top P^{1*}_{t+1} B_{t}^{1} & \dots & (B_{t}^{1})^\top P^{1*}_{t+1} B^N_t\\
    (B_{t}^{2})^\top P^{2*}_{t+1} B^1_t & \dots &  (B_{t}^{2})^\top P^{2*}_{t+1} B^N_t\\
    \vdots & \ddots & \vdots \\
    (B_{t}^{N})^\top P^{N*}_{t+1} B^1_t & \dots &  R_{t}^{N} + (B_{t}^{N})^\top P^{N*}_{t+1} B_{t}^{N}
    \end{bmatrix}, \\
     &\Psi_t = \begin{bmatrix}
        -F_t^\top P_{t+1}^{1*} B_t^1 + K_{t}^{1*}R_t^1 & \dots & -F_t^\top P_{t+1}^{1*} B_t^N + K_{t}^{N*}R_t^N\\
        -F_t^\top P_{t+1}^{2*} B_t^1 + K_{t}^{1*}R_t^1 & \dots & -F_t^\top P_{t+1}^{2*} B_t^N + K_{t}^{N*}R_t^N\\
        \vdots & \ddots & \vdots \\
        -F_t^\top P_{t+1}^{N*} B_t^1 + K_{t}^{1*}R_t^1 &\dots & -F_t^\top P_{t+1}^{N*} B_t^N + K_{t}^{N*}R_t^N
    \end{bmatrix}, \\
& \Phi_t  \in \bbR^{Nn_u \times Nn_u}, \quad \Psi_t  \in  \mathbb{R}^{Nn_x \times N n_u},
\end{align*}
}

{\small
\begin{align*}
    &H_t = \;\Phi_t^{-1} \begin{bmatrix} (B_{t}^{1})^\top & 0 & \dots & 0\\
        0 &  (B_{t}^{2})^\top & \cdots & 0 \\
        \vdots & \vdots & \ddots & \vdots \\
        0 & 0 & \cdots & (B_{t}^{N})^\top
    \end{bmatrix} \in  \mathbb{R}^{Nn_u \times N n_x}, \\
    &H = \begin{bmatrix}
         H_{1} & \cdots & 0 \\
         \vdots & \ddots & \vdots \\
         0 & \cdots & H_{T}
    \end{bmatrix} \in  \mathbb{R}^{TNn_u \times TNn_x}, \\
    &C_t = \underbrace{\begin{bmatrix}
        c_t \\ \vdots \\ c_t
    \end{bmatrix}}_{N\text{ blocks}} \in \bbR^{N n_x \times O}, \vc = \begin{bmatrix}
        0^{n_x \times O}  \\
        c_1\\
        \vdots  \\
        c_T
    \end{bmatrix} \in \mathbb{R}^{(T+1)n_x \times O},\\
    &F_t = A_t - \sum_{i=1}^N B_t^i K_t^{i*} \in \mathbb{R}^{n_x \times n_x}, \\
    &F = \begin{bmatrix} I, &
        F_0, &
        F_1F_0, &
        \dots, &
        \prod_{t=T-1}^{0} F_t
    \end{bmatrix}^\top \in  \mathbb{R}^{(T+1)n_x \times n_x}, \\
    &B_t = \begin{bmatrix}
        B_t^1 & B_t^2 & \cdots & B_t^N
    \end{bmatrix} \in \mathbb{R}^{n_x \times N n_u}, \\
    &Z = \begin{bmatrix}
        0 & 0  & \cdots & 0\\
        -B_0 & 0 & \cdots & 0\\
        -F_1 B_0 & -B_1 & \cdots & 0\\
        -F_2 F_1 B_0  & -F_2 B_1 & \cdots & 0\\
        \vdots & \vdots & \ddots & \vdots \\
        -\prod_{t=T-1}^{1} F_t \cdot B_0 & -\prod_{t=T-1}^{2} F_t \cdot B_1 & \cdots & -B_{T-1}
    \end{bmatrix}, \\
    &Z \in \mathbb{R}^{(T+1) n_x \times T N n_u}, \\
    &G_T = C_T \in \bbR^{N n_x \times O},\\
    &G_t =  \operatorname{diag}(\underbrace{F_t,\ldots,F_t}_{N})  G_{t+1} + \Psi_t H_t G_{t+1} + C_t \in \bbR^{N n_x \times O}, \\ 
    &G = \begin{bmatrix} G_1 & \dots & G_T
    \end{bmatrix}^\top \in  \mathbb{R}^{TNn_x \times O}, \\
    &\cQ^i = \begin{bmatrix}
        0 & 0 & \cdots & 0 \\
        0 & Q_1^i & \cdots & 0 \\
        \vdots & \vdots & \ddots & \vdots \\
        0 & 0 & \cdots & Q_T^i
    \end{bmatrix} \in \mathbb{R}^{(T+1)n_x \times (T+1)n_x}, \\
    &\cR^i = \begin{bmatrix}
        R_0^i & 0 & \cdots & 0 \\
        0 & R_1^i & \cdots & 0 \\
        \vdots & \vdots & \ddots & \vdots \\
        0 & 0 & \cdots & R_{T-1}^i
    \end{bmatrix} \in \mathbb{R}^{Tn_u \times Tn_u}. 
\end{align*}}

\end{document}